\documentclass[lettersize,journal]{IEEEtran}
\usepackage{amsmath,amsfonts}
\usepackage{array}
\usepackage[caption=false,font=normalsize,labelfont=sf,textfont=sf]{subfig}
\usepackage{textcomp}
\usepackage{stfloats}
\usepackage{url}
\usepackage{verbatim}
\usepackage{graphicx}
\def\BibTeX{{\rm B\kern-.05em{\sc i\kern-.025em b}\kern-.08em
    T\kern-.1667em\lower.7ex\hbox{E}\kern-.125emX}}
\usepackage{balance}
\usepackage{cite}
\usepackage{stmaryrd}
\usepackage[ruled,vlined]{algorithm2e}
\usepackage{tcolorbox}
\tcbuselibrary{listings}

\usepackage{algorithmic}
\usepackage{algorithm2e}
\usepackage{pifont}
\usepackage{amsthm}
\usepackage[colorlinks=true, linkcolor=blue]{hyperref}

\usepackage{amssymb, tikz}
\usetikzlibrary{shapes}
\usepackage{framed}
\usepackage{tcolorbox}
\usepackage{enumitem}

\usepackage{booktabs} 
\usepackage{multirow}
\usepackage{xcolor}      
\usepackage{colortbl}

\newcommand{\ieeediamond}[1]{%
\begin{tikzpicture}[baseline=-0.65ex]
\node[diamond, draw=black, fill=white, minimum size=5mm, inner sep=0pt] {#1};
\end{tikzpicture}
}

\newtheorem{theorem}{Theorem}
\newtheorem{definition}{Definition}
\newtheorem{proposition}{Proposition}

\begin{document}
\title{High-Throughput and Scalable Secure Computation Protocols for Deep Learning with Packed Secret Sharing}
\author{Qinghui Zhang, Xiaojun Chen, Yansong Zhang, Xudong Chen, and Tingyu Fan
\thanks{Qinghui Zhang, Xiaojun Chen, Yansong Zhang, Xudong Chen, and Tingyu Fan are with the Institute of Information Engineering, Chinese Academy of Sciences, Beijing, China; the State Key Laboratory of Cyberspace Security Defense, Beijing, China; and also with the School of Cyber Security, University of Chinese Academy of Sciences, Beijing, China (e-mail: zhangqinghui@iie.ac.cn;chenxiaojun@iie.ac.cn;zhangyansong@iie.ac.cn;
chenxudong@iie.ac.cn;fantingyu@iie.ac.cn).}}

\maketitle

\begin{abstract}
Most existing privacy-preserving deep learning protocols based on secure multi-party computation (MPC) typically support at most four participants, demonstrating severely limited scalability. Liu et al. (USENIX Security'24) presented the first relatively practical approach by utilizing Shamir secret sharing with Mersenne prime fields. However, when processing deeper neural networks such as VGG16, their protocols incur substantial communication overhead, resulting in particularly significant latency in wide-area network (WAN) settings. To tackle this challenge, in this paper, we propose a high-throughput and scalable MPC protocol for deep learning in the honest-majority setting. The core of our approach lies in leveraging packed Shamir secret sharing (PSS) to enable parallel computation and reduce communication complexity. The main contributions are three-fold: i) We define a primitive called vector-matrix multiplication-friendly random share tuples and present a communication-efficient protocol for vector–matrix multiplication based on the primitive. ii) Inspired by the inherent parallelizability of multiple-input multiple-output convolutions, we design the filter packing approach that enables parallel convolution. iii) We further extend all non-linear protocols based on Shamir Secret Sharing to their PSS variants, which enable parallel non-linear evaluations. Extensive experiments across various datasets and neural networks demonstrate the superiority of our approach in WAN. Specifically, compared to Liu et al. (USENIX Security'24), our scheme reduces the communication by up to $5.9\times$, $11.2\times$, and $6.8\times$ in offline, online and total communication overhead, respectively. In addition, our scheme is up to $1.6\times$, $2.6\times$, and $1.8\times$ faster in offline, online and total running time, respectively.

\end{abstract}


\section{Introduction}
\label{Introduction}
\IEEEPARstart{D}{eep} learning (DL) has emerged as a transformative methodology with widespread applications across diverse domains, particularly in computer vision and natural language processing. However, its reliance on large-scale data, which often includes sensitive information, raises significant privacy concerns that require careful consideration. Secure multi-party computation (MPC) offers a promising solution to this challenge by enabling a set of $n$ parties to jointly compute a function over their private inputs while preserving both input privacy and output correctness~\cite{Yao1982ProtocolsFS,Goldreich1987HowTP,Chaum1988MultipartyUS,BenOr1988CompletenessTF}. Integrating MPC with deep learning thus allows multiple parties to collaboratively perform training or inference tasks without exposing their private data.

In recent years, numerous studies have been conducted to enhance the communication and computational efficiency of MPC-based privacy-preserving deep learning (PPDL) protocols. In the dishonest-majority setting, most existing protocols focus on the two-party scenario~\cite{Mohassel2017SecureMLAS,Demmler2015ABYA,Patra2020ABY20IM,Juvekar2018GazelleAL,Liu2017ObliviousNN,Riazi2018ChameleonAH,Chandran2019EzPCPA,Rathee2020CrypTFlow2P2,Lehmkuhl2021MUSESI,Xu2022SIMC2I}. In the honest-majority setting, most existing protocols consider the three-party scenario~\cite{Mohassel2018ABY3AM,Wagh2019SecureNN3S,Wagh2020FalconHM,Patra2020BLAZEBF,Li2023Efficient3F,Koti2020SWIFTSA} and four-party scenario~\cite{Byali2020FLASHFA,Koti2021TetradAS,Rachuri2019TridentE4,Dalskov2020FantasticFH}, where the protocols tolerate one corrupted party. Unfortunately, all the aforementioned protocols universally suffer from scalability limitations. This critical shortcoming was recently addressed by Liu et al.~\cite{Liu2024ScalableMC}, who designed a scalable protocol that can support a large number of parties for PPDL based on Shamir secret sharing~\cite{shamir1979share} against semi-honest adversaries in the honest-majority setting. Nevertheless, when we rerun the scalable protocol on relatively deeper neural networks such as VGG16, it exhibits poor performance owing to significant communication overhead. Hence, it is necessary to develop a more communication-efficient protocol that can better accommodate the requirements of real-world scenarios. For ease of reference, we denote this work~\cite{Liu2024ScalableMC} as “LXY24”.

In theory, Franklin and Yung introduced a parallelizable technique called packed Shamir secret sharing (PSS)~\cite{franklin1992communication}, which packs $k$ secrets into a single secret sharing. Naively replacing the Shamir secret sharing scheme in the LXY24 with a PSS scheme can enable parallel processing of $k$ instances,  which results in an amortized reduction in communication overhead by a factor of approximately $k$. However, in certain scenarios, such as a photo-based recognition application, only a single input may be available. If the above packed approach across multiple instances is still employed, it needs to introduce dummy inputs and pack them together with the real input to complete the inference process. More critically, this approach degenerates into the protocol similar to LXY24, thereby completely forfeiting the advantage of parallel computing provided by PSS. 

\textit{Given that the naive method isn’t efficient for processing a single input, we can turn to employ the substantial parallelism within a single inference pass.} To be specific, deep neural networks typically consist of linear layers and non-linear layers. Linear layers primarily execute matrix operations. For example, the fully connected layer is a vector-matrix multiplication. The convolutional layer can be expressed as a matrix multiplication~\cite{Wagh2019SecureNN3S}, which can be achieved through multiple vector-matrix multiplications. All these operations can be efficiently implemented using the blocked matrix strategy, which is particularly amenable to parallel computing acceleration. Under the block-processing strategy, each block of the vector and matrix is encoded into an individual PSS value, and block-wise multiply-accumulate operations are performed using DN protocol based on PSS (a natural extension of the DN protocol in~\cite{damgaard2007scalable}). Fig.~\ref{sh} shows a toy example. Particularly, ~\cite{goyal2022sharing} proposes the sharing transformation protocol that can transform a PSS packing secrets $(x_1,x_2)$ into one packing  $(x_1+x_2, 0)$ for the case $k=2$. Regarding non-linear layers, such as the ReLU function defined as $\text{ReLU}(x)=\text{Max}(0,x)$, the operations are element-wise. This means each output value depends solely on the single value at the corresponding position in the input tensor. Hence, the parallel non-linear operations are entirely achievable by packing $k$ values into an individual PSS value. \textit{Nevertheless, the above simple exploitation of parallelism in a single inference pass is not practical, especially for linear layers.}

\begin{figure*}[t]
    \centering
    \includegraphics[width=1\linewidth]{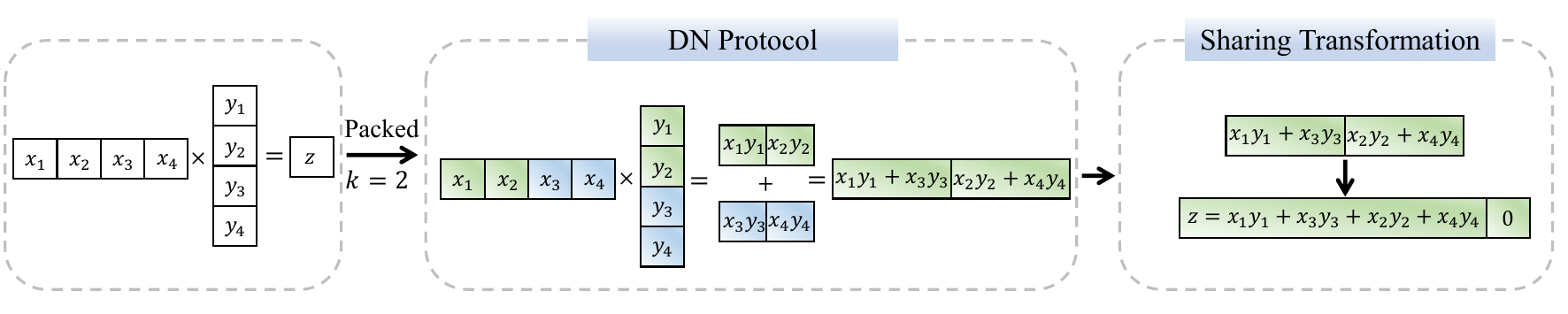}
    \caption{A toy example. Boxes with the same color represent a packed secret sharing, in which $k$ values are packed.}
    \label{sh}
\end{figure*}

\textit{First, the trivial vector-matrix multiplication eliminates the benefit of using the PSS.} As depicted in Fig.~\ref{sh}, each block of the vector is encoded into an individual PSS value, and block-wise multiply-accumulate operations are performed using the DN protocol. Since the computed result remains in a partitioned form, this necessitates an interactive sharing transformation protocol to combine these partial results into the final output. Compared with  LXY24, the entire process instead incurs more communication overhead. The reason is that LXY24 only involves the DN protocol for vector-matrix multiplication.

\textit{Second, when processing convolutional layers, zero padding of convolutional inputs incurs significant communication cost.} Recall that zero padding is also frequently applied to the convolutional inputs in many neural networks, such as VGG16. A convolutional layer can be achieved through multiple vector-matrix multiplications. Unfortunately, when zero padding needs to be applied to convolutional inputs, it becomes necessary to combine all inputs with these zero padding values into new PSS values. This combination process incurs significant communication overhead.

\subsection{Technique Overview}
\textit{To address the first issue,} we define a primitive called vector-matrix multiplication-friendly random share tuples (VM-RandTuple), which can be generated in batches by employing the Vandermonde matrix in the offline phase. Recall that the DN protocol~\cite{damgaard2007scalable} consists of two steps: First, the parties locally compute and send the results to $P_1$. Second, $P_1$ reconstructs the result and sends it back to the parties. Our optimization is motivated by the fact that $P_1$ holds the partial results plaintext, and the final output arises as a composition of partial results and VM-RandTuple, enabling the interactive sharing transformation protocol to be delegated and locally implemented within $P_1$. Theoretically, our optimization has the potential to achieve better online communication complexity vector-matrix multiplication.

\textit{To address the second issue,} we observed that the input is convolved with each filter in either a single-input multiple-output or multiple-input multiple-output convolution process. Therefore, such a convolution process inherently involves parallel computation. In light of this, we propose a filter packing approach, i.e., we pack $k$ values that are selected respectively from the same position of $k$ filters into a single packed secret share. Correspondingly, we also pack $k$ copies of each input pixel into a single packed secret share. Based on this packing approach, not only do we achieve a
parallel convolution computation, but also there is no need to recombine zero padding values with the input to form new PSS representations. Instead, we directly pack $k$ copies of zero padding values locally. Moreover, due to $k$ copies of the input being packed into a single packed secret share, which differs from the packing approach of the vector-matrix multiplication, we need to design a protocol for the inter-layer connectivity.

\subsection{Contributions}
 Building upon the above techniques, in this paper, we propose a parallel and scalable MPC protocol by leveraging PSS for PPDL. The main contributions are summarized below.
 
\begin{itemize}
    \item {\textbf{Efficient Vector-matrix Multiplication.} We define a primitive called vector-matrix multiplication-friendly random share tuples and design a protocol to generate these tuples for computing vector-matrix multiplication.  We further combine these tuples with the DN~\cite{damgaard2007scalable} protocol, resulting in an efficient protocol for vector-matrix multiplication with a single round of online communication.}
    \item {\textbf{Parallel Convolution Computation.} We propose a filter packing approach for convolution computation. By integrating our method with the DN protocol, we achieve a parallel convolution computation protocol that requires merely single-round online communication. Our convolution computation protocol can further support zero padding without extra communication.}
    \item {\textbf{Parallel Non-linear Operations.} We extend all protocols based on Shamir secret sharing for non-linear operations in LXY24 to the PSS-based protocols, which enable parallel computation of the non-linear operations, resulting in efficient protocols for parallel ReLU with $O(\text{log}_2{\ell})$ rounds of online communication.}
    \item {\textbf{Evaluations.} We implement all protocols and conduct extensive evaluations on various neural networks in both LAN and WAN settings with varying numbers of parties and different packaging parameters. Experimental results demonstrate progressively enhanced performance advantages with increasing numbers of participating entities and deeper neural network architectures. Compared to LXY24, our scheme reduces the communication by up to $5.9\times$, $11.2\times$, and $6.8\times$ in offline, online, and total communication overhead, respectively. At the same time, our scheme is up to $1.6\times$, $2.6\times$, and $1.8\times$ in offline, online and total running time, respectively. Furthermore, experimental results reveal the superiority of our protocol in terms of memory overhead. For example, in our experiment setup, our protocol reliably supports VGG16 inference, whereas LXY24 suffers from memory overflow problems.}  
\end{itemize}

\section{Preliminary}
\subsection{Notation}
Let $p$ be a Mersenne prime, which is the form of $p=2^{\ell}-1$ for some prime $\ell$. In PPDL, $\ell$ is chosen to be 31 and 61, which is sufficient. $\mathbb{F}_p$ denotes a finite field. We denote by $[n]$ the set $\{0,1,\cdots,n-1\}$, where $n\in\mathbb{N},\mathbb{N}$ is the set of natural numbers. We use the lowercase letter $x$ to denote a scalar and $x_i$ to denote the $i$-th bit of the scalar $x$. The bold lowercase letter $\boldsymbol{x}=\{\boldsymbol{x_i}\}_{i=0}^{n-1}$ is a vector, $\boldsymbol{x}_i$ is the $i$-th element of the vector $\boldsymbol{x}$, $\boldsymbol{x}_{i,j}$ is $j$-th bit of the $i$-th element of the vector $\boldsymbol{x}$ and $\boldsymbol{x}^i$ is indexed by $i$ to distinguish different vectors. Correspondingly, uppercase letter $X$ is a matrix and $X_{i,j}$ is the element indexed by row $i$ and column $j$ of the matrix $X$. The bold uppercase letter $\boldsymbol{X}$ is a tensor. $Van(n, n-t)$ is an $n$-by-$(n-t)$ Vandermonde matrix.

\subsection{Packed Shamir Secret Sharing Scheme}
\label{PSS}
The packed Shamir secret sharing (PSS) scheme~\cite{franklin1992communication} is a natural generalization of the standard Shamir secret sharing scheme~\cite{shamir1979share}.

For a vector $\boldsymbol{x}\in \mathbb{F}_q^k$, we will use $\llbracket \boldsymbol{x} \rrbracket_d\in\mathbb{F}_p$ to denote a degree-$d$ PSS, where $k-1\leq d\leq n-1$. It requires $d+1$ shares to reconstruct the whole secrets, and any $d-k+1$ shares are independent of the secrets. The scheme $\llbracket \boldsymbol{x} \rrbracket_d$ corresponds to a degree-$d$ polynomial $f$ such that the values $\{f(i)\}_{i=1}^{n}$ represent the shares and $\{f(\boldsymbol{s}_i)=\boldsymbol{x}_i\}_{i=0}^{k-1}$, where $\{\boldsymbol{s}_i\}_{i=0}^{k-1}$ are the default positions to
store the secrets $\boldsymbol{x}=\{\boldsymbol{x}_i\}_{i=0}^{k-1}$. The PSS scheme has the following nice properties: 
\begin{itemize}
    \item Linear Homomorphism: For all $d\geq k-1$ and $\boldsymbol{x},\boldsymbol{y}\in \mathbb{F}_q^k,\llbracket \boldsymbol{x} \rrbracket_d+\llbracket \boldsymbol{y} \rrbracket_d=\llbracket \boldsymbol{x}+\boldsymbol{y} \rrbracket_d$. Assume that $c$ is a public constant, $c+\llbracket \boldsymbol{x} \rrbracket_d$ and $c\times\llbracket \boldsymbol{x} \rrbracket_d$ denote adding $c$ to each element of $\boldsymbol{x}$ and multiplying each element of $\boldsymbol{x}$ by $c$, respectively.
    \item Multiplication-friendliness in~\cite{goyal2022sharing}: When $d\leq n-k$, all parties can locally multiply a public vector $\boldsymbol{c}\in\mathbb{F}_p^k$ with the share $\llbracket \boldsymbol{x} \rrbracket_{d}$. First, all parties locally compute $\llbracket \boldsymbol{c} \rrbracket_{k-1}$. Then all parties can locally compute $\llbracket \boldsymbol{c}* \boldsymbol{x} \rrbracket_{k-1+d}=\llbracket \boldsymbol{c} \rrbracket_{k-1}\times\llbracket \boldsymbol{x} \rrbracket_{d}$, where $*$ denote coordinate-wise multiplication.
     \item Multiplicative: For all $d_1,d_2\geq k-1,d_1+d_2 < n$ and for all $\boldsymbol{x},\boldsymbol{y}\in \mathbb{F}_q^k,\llbracket \boldsymbol{x} \rrbracket_{d_1}\times\llbracket \boldsymbol{y} \rrbracket_{d_2}=\llbracket \boldsymbol{x}*\boldsymbol{y} \rrbracket_{d_1+d_2}$. If $d_1=d_2=\frac{n-1}{2}$, we can compute $\llbracket \boldsymbol{x} \rrbracket_{d}*\llbracket \boldsymbol{y} \rrbracket_{d}=\llbracket \boldsymbol{x}\times\boldsymbol{y} \rrbracket_{d}$ by a natural extension of the DN multiplication protocol in~\cite{damgaard2007scalable}. We summarize the functionality $\mathcal{F}_{PDN}$ in Appendix~\ref{appendixA}. The protocol is as follows:
    \begin{list}{}{}
        \item {(1) There are $n$ parties $\{P_i\}_{i=1}^{n}$. They prepare a pair of two random sharings ($\llbracket \boldsymbol{r} \rrbracket_{2d},\llbracket \boldsymbol{r} \rrbracket_{d}$) of the same secrets vector $\boldsymbol{r}\in\mathbb{F}_p^k$.}
        \item {(2) All parties compute $\llbracket \boldsymbol{a} \rrbracket_{2d}=\llbracket \boldsymbol{x} \rrbracket_{d}\times\llbracket \boldsymbol{y} \rrbracket_{d}+\llbracket \boldsymbol{r} \rrbracket_{2d}$ and send $\llbracket \boldsymbol{a} \rrbracket_{2d}$ to $P_1$.}
        \item {(3) $P_1$ reconstructs the secrets $\boldsymbol{a}$, generates $\llbracket \boldsymbol{a} \rrbracket_{d}$, and distributes the shares to all other parties.}
        \item {(4) All parties compute $\llbracket \boldsymbol{z} \rrbracket_{d}=\llbracket \boldsymbol{a} \rrbracket_{d}-\llbracket \boldsymbol{r} \rrbracket_{d}$.}
    \end{list}
\end{itemize}

In addition, a degree-$t$ Shamir secret sharing scheme corresponds to a degree-$t$ polynomial $f$ such that $\{f(i)\}_{i=1}^{n}$  are the shares and $f(0)$ is the secret. However, we do not always need to store the secret at point 0. We use $[x|_j]_t$ to denote a degree-$t$ Shamir sharing whose secret is stored at the point $j,j\notin\{1,\cdots,n\}$.

Now let $E_i\in\{0,1\}^k$ denote the unit vector, whose $i$-th entry is 1 and all other entries are 0. Then parties can convert $E_i$ to $\llbracket E_i\rrbracket_{k-1}$ locally. There are the following nice properties:
\begin{itemize}
    \item Transforming Shamir secret shares at $k$ distinct positions into a single PSS locally: 
    \begin{center}
        $\llbracket \boldsymbol{x} \rrbracket_{t+k-1}=\sum_{i=0}^{k-1}\llbracket E_i \rrbracket_{k-1}\cdot[\boldsymbol{x}_i|_{s_i}]_t$.
    \end{center}
    \item Correspondingly, we can choose $k$ secrets from different position of $k$ distinct PSS, and generates a new PSS that the positions of the secrets we choose are from different positions remain unchanged in the new PSS: 
    \begin{center}
    $\llbracket \boldsymbol{x} \rrbracket_{d+k-1}=\sum_{i=0}^{k-1}\llbracket E_{i} \rrbracket_{k-1}\cdot\llbracket \boldsymbol{x}^i\rrbracket_{d}$, 
    \end{center}
    where $\boldsymbol{x}=\{\boldsymbol{x}^i_i\}_{i=0}^{k-1}$. 
\end{itemize}

\subsection{Neural Network}
\label{NN}
In this paper, we focus on the convolutional neural networks (CNN), a widely used in computer vision. It mainly consists of fully connected layers, convolutional layers and non-linear layers.

\textbf{Fully Connected Layer (FC).} The input and output of FC are vector $\boldsymbol{v}\in\mathbb{F}_p^{n_i}$ and $\boldsymbol{u}\in\mathbb{F}_p^{n_0}$ respectively. An FC is formalized as a tuple $(W,\boldsymbol{b})$ where $W$ is a $(n_i\times n_o)$ weight matrix and $\boldsymbol{b}$ is a bias vector of length $n_o$. The output is represented as a linear transformation $\boldsymbol{u}=\boldsymbol{v}W+\boldsymbol{b}$. In a word, the FC is a vector-matrix multiplication. We propose
an efficient protocol for computing vector-matrix multiplication in \S\ref{sec::Vector-Matrix Multiplication}.

\textbf{Convolution Layer (Conv).} The input to a Conv layer is a tensor $\boldsymbol{X}\in \mathbb{R}^{w_i\times h_i\times c_i}$ where $w$ is the image width, $h$ is the image height, and $c$ is the number of input channels. The Conv is then parameterized by the tuple $(f_w,f_h,c_i,c_o)$. In other words, there are $c_o$ filter banks, each consisting of $c_i$ many $f_w\times f_h$ filters. The Conv process is sliding each filter with a
certain stride $s$ on the tensor $\boldsymbol{X}$. The output of the Conv can then be parameterized by the tuple $(w_o,h_o,c_o)$ where $w_o=\lfloor (w_i-f_w+2p)/s\rfloor+1, h_o=\lfloor (h_i-f_h+2p)/s\rfloor+1, p$ is padding size applied to the edges of each $w_i\times h_i$ matrix of the tensor $\boldsymbol{X}$. Fig.~\ref{CNNcompute} shows the Conv process. For efficiently achieving convolution computation, we design an efficient protocol in \S\ref{sec::Conv}.
\begin{figure}[t]
    \centering
    \includegraphics[width=0.5\textwidth]{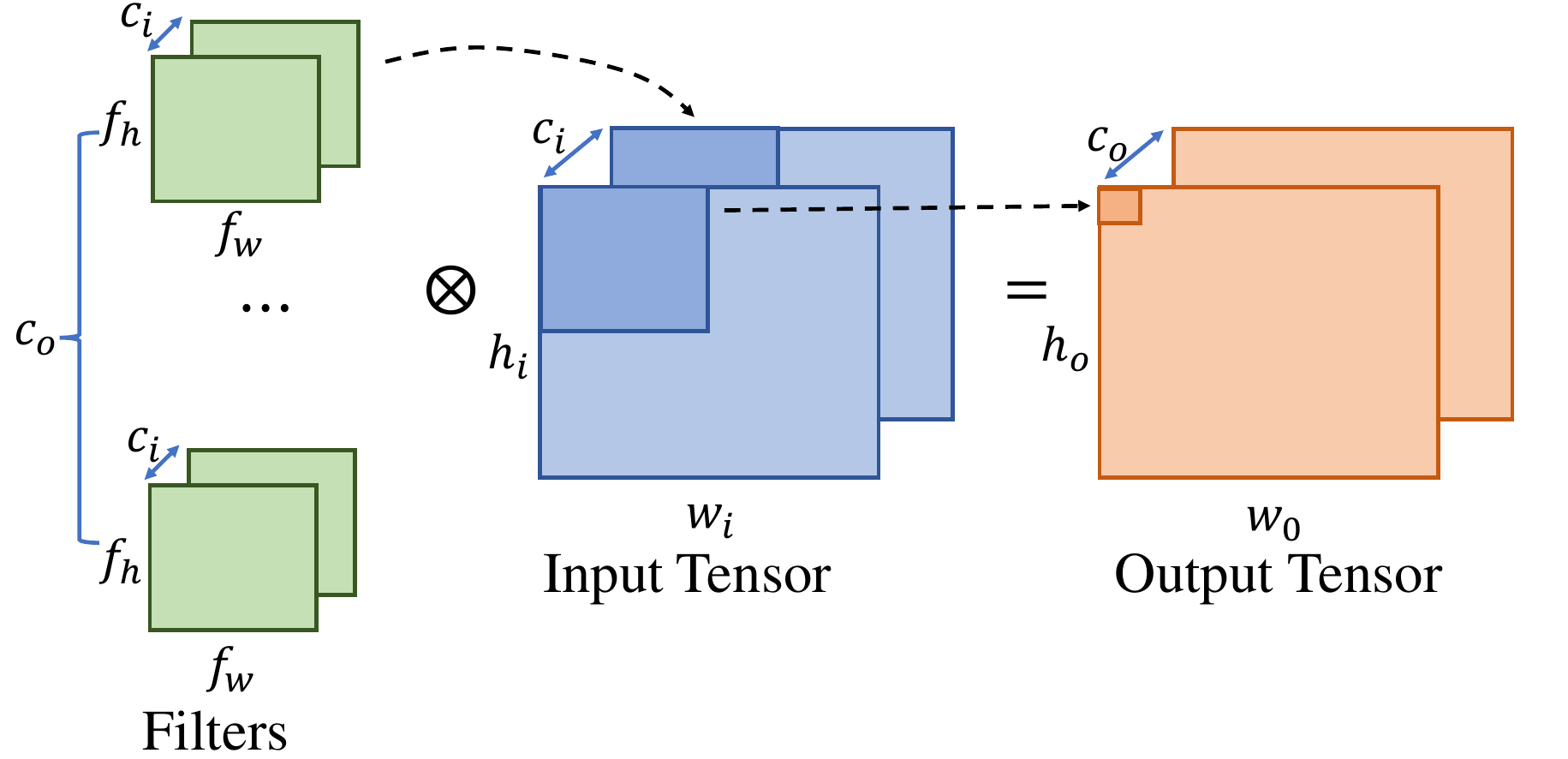}
    \caption{Convolution process.}
    \label{CNNcompute}
\end{figure}

\textbf{Non-Linear Layer (ReLU and Pool).} In deep learning, the non-linear layers consist
of activation functions and pooling functions. The activation functions are applied to each element of input independently. The common activation functions is ReLU function: $\text{ReLU}(x)=\text{Max}(0,x)$. The pooling function reduces the output size. One of the most popular pooling functions in CNN is the MaxPool function that arranges inputs into several windows and calculates the maximum for each window. We construct the parallel non-linear operators in \S\ref{sec::nonline}.

\subsection{Model}
 We assume there are $n\in\mathbb{N}$ parties, denoted by $P_1, P_2,\cdots,P_n,$ where $n\geq 2d+1,d\in\mathbb{N}$. We describe our protocols in the semi-honest model with passive security. For every pair of parties, assume there exists a private and authentic synchronous channel. The number of corrupted parties that are assumed to be poly-time bounded is at most $t=d-k+1$ for a natural extension of the DN~\cite{damgaard2007scalable} multiplication protocol, where $2\leq k\leq d$. If $k=1$, the PSS becomes standard Shamir secret sharing. It is worth noting that a degree-$d$ PSS can pack $k\in\mathbb{N}$ secrets, and $k$ and default positions $\{\boldsymbol{s}_i\}_{i=0}^{k-1}$ are publicly known to the $n\in\mathbb{N}$ parties.

\section{Efficient Vector-Matrix Multiplication}
\label{sec::Vector-Matrix Multiplication}
In this section, we first present the design of the vector-matrix multiplication protocol in \S\ref{Protocol-Design-VM}. Then, we extend the vector-matrix multiplication protocol to support truncation in \S\ref{trun}.

\begin{figure*}[t]
    \centering
    \includegraphics[width=1\linewidth]{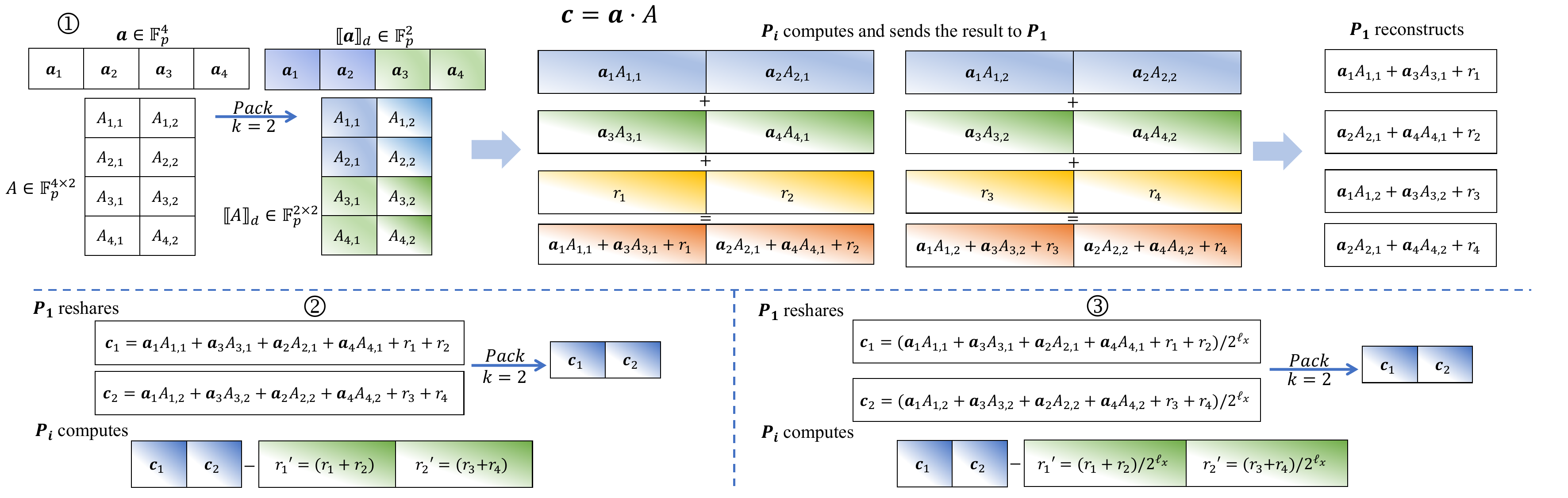}
    \caption{A toy example. Boxes with the same color represent a packed secret sharing, in which 
$k$ values are packed and located at positions $(\boldsymbol{s}_0,\cdots,\boldsymbol{s}_{k-1})$; blank boxes indicate plaintext values. \ding{172} and \ding{173} correspond to general vector-matrix multiplication. \ding{172} and \ding{174} correspond to fixed-point vector-matrix multiplication.}
    \label{multtoy}
\end{figure*}

\subsection{Protocol Design}
\label{Protocol-Design-VM}
FC is a key layer in the deep learning \S\ref{NN}, which consists of vector-matrix multiplications. In light of the earlier analysis, vector-matrix multiplication $\boldsymbol{a}\cdot A$ can be calculated by blocked matrix strategy, where $\boldsymbol{a}\in\mathbb{F}_{p}^u,A\in\mathbb{F}_{p}^{u\times v}$. Specifically, we divide the vector $\boldsymbol{a}\in\mathbb{F}_{p}^u$ into $\lceil u/k\rceil$ blocks and divide each column of the matrix $A\in\mathbb{F}_{p}^{u\times v}$ into $\lceil u/k\rceil$ blocks. The size of each block is $k$. Then we pack each block into a PSS (e.g., $\llbracket \boldsymbol{a}\rrbracket\in\mathbb{F}_p^{\lceil u/k\rceil},\llbracket A\rrbracket\in\mathbb{F}_p^{\lceil u/k\rceil\times v}$) and combine it with the DN~\cite{damgaard2007scalable} protocol. 
 
In order to reduce communication cost, we first define vector-matrix multiplication-friendly random share tuple (VM-RandTuple), i.e., $(\llbracket \boldsymbol{r}\rrbracket_{2d}\in\mathbb{F}_{p}^{k},\llbracket \boldsymbol{r'}\rrbracket_{d}\in\mathbb{F}_{p})$, where $\boldsymbol{r}=\{\boldsymbol{r}_i\}_{i=0}^{k^2-1},\boldsymbol{r}'=\{\boldsymbol{r}'_i=\sum_{j=i\cdot k}^{(i+1)\cdot k-1}\boldsymbol{r}_j\}_{i=0}^{k-1}$. We can utilize a Vandermonde matrix for the batch generation of such random share tuples in the offline phase. We formally define the functionality $\mathcal{F}_{VM-RandTuple}$ in Appendix~\ref{appendixA} and provide a secure protocol implementation in Algorithm~\ref{RP}. Based on the VM-RandTuple, a key optimization to the second step of DN protocol is having $P_1$ combine these partial results into the final output and reshare the final output to the parties. This optimization directly leads to a $k$-fold reduction in the amortized communication cost of the second step of DN protocol, compared with LXY24. 

We describe a toy example in Fig.~\ref{multtoy}. It embodies a core idea of our technology. \ding{172} and \ding{173} correspond to general vector-matrix multiplication; \ding{172} and \ding{174} correspond to fixed-point vector-matrix multiplication. We defer the details of the latter to \S\ref{trun}. In Fig.~\ref{multtoy}, we set $k=2$, meaning that $2$ values are packed into a single packed secret share. Accordingly, the vector $\boldsymbol{a}$ is packed into 2 packed secret shares, i.e., $\llbracket\boldsymbol{a}\rrbracket\in\mathbb{F}_p^2$ and the matrix $A$ is packed into $2\times2$ packed secret shares, i.e., $\llbracket A\rrbracket\in\mathbb{F}_p^{2\times2}$. Then all parties $\{P_i\}_{i=1}^{n}$ can compute $\llbracket \boldsymbol{a}\cdot A \rrbracket_{2d}=\llbracket \boldsymbol{a} \rrbracket_{d}\cdot\llbracket A \rrbracket_{d}$ locally. Since $\llbracket \boldsymbol{a}\cdot A \rrbracket_{2d}$ has a degree of $2d$, a degree reduction is required. $P_i$ sends the $\llbracket \boldsymbol{a}\cdot A \rrbracket_{2d}+\llbracket\boldsymbol{r}\rrbracket_{2d}$ to $P_1$ and $P_1$ reconstructs the secret values at the default positions $\{\boldsymbol{s}_i\}_{i=0}^{k-1}$ and then calculates the sum of the results. Finally, $P_1$ packs the multiple results, i.e., $(\boldsymbol{c}_1,\boldsymbol{c}_2)$ into a single degree-$d$ PSS and shares to the other parties. We formally define the functionality $\mathcal{F}_{VecMatMult}$ in Appendix~\ref{appendixA} and provide a secure protocol implementation in Algorithm~\ref{MVMP}. 

To sum up, given a vector of length $u$ and a $u$-by-$v$ matrix, the online complexity is exactly 1 round with $(1+\frac{1}{k})uv$ field elements amortized per party. The offline complexity is exactly 2 round with $(1+\frac{1}{k})\frac{2nv}{n+2k-1}$ field elements amortized per party.

\begin{algorithm}
\caption{VM-RandTuple $\Pi_{VM-RandTuple}$}
\KwIn{None.}
\KwOut{ $\llbracket \boldsymbol{r}\rrbracket_{2d}\in\mathbb{F}_p^{(n-t)\cdot k},\llbracket\boldsymbol{r'}\rrbracket_d\in\mathbb{F}_p^{(n-t)}$, where 

$\boldsymbol{r}=\{\boldsymbol{r}_i\}_{i=0}^{(n-t)\cdot k^2-1}, $
$\boldsymbol{r}'=\{\boldsymbol{r}'_i=\sum_{j=i\cdot k}^{(i+1)\cdot k-1}\boldsymbol{r}_j\}_{i=0}^{(n-t)\cdot k -1}.$}

    \begin{enumerate}[label=\arabic*., align=left, leftmargin=1.2em, labelwidth=1em, labelsep=0.2em]
        \item $P_i,i\in\{1,\cdots,n\}$ randomly samples a vector $\boldsymbol{a}^{i}$ with length of $k^2$ and distributes $\llbracket\boldsymbol{a}^{i}\rrbracket_{2d}\in\mathbb{F}_p^k.$
        \item $P_i$ computes $\boldsymbol{b}^{i}=\{\boldsymbol{b}_u^i=\sum_{j=u\cdot k}^{(u+1)\cdot k-1}\boldsymbol{a}^i_j\}_{u=0}^{k-1}$ and distributes $\boldsymbol{c}^{i}=\{[\boldsymbol{b}^{i}_{u}|_{\boldsymbol{s}_u}]_t\}_{u=0}^{k-1}.$
        \item $P_i$ computes 

            \begin{center}
                  $A =Van(n,n-t)^T\cdot(\llbracket \boldsymbol{a}^{1}\rrbracket_{2d},\cdots,\llbracket \boldsymbol{a}^{n}\rrbracket_{2d} ).$
            
             $B =Van(n,n-t)^T\cdot(\boldsymbol{c}^{1},\cdots,\boldsymbol{c}^{n}).$
            \end{center}

        \item \For{$j \leftarrow 0$ \KwTo $n-t-1$}{
            Compute $\boldsymbol{f}_{j} =\sum_{i=0}^{k-1}\llbracket E_i \rrbracket_{k-1}\cdot B_{j,i}.$
            
            \For{$u \leftarrow 0$ \KwTo $k-1$}{
            Set $\boldsymbol{g}_{j\cdot k+u} =A_{j,u}.$
            
            }
        }
       
        \item $P_i$ sets $\llbracket\boldsymbol{r}\rrbracket_{2d}=\boldsymbol{g}=\{\boldsymbol{g}_i\}_{i=0}^{(n-t)\cdot k},\llbracket\boldsymbol{r}'\rrbracket_{d}=\boldsymbol{f}=\{\boldsymbol{f}_i\}_{i=0}^{n-t-1}.$
    \end{enumerate}
\Return $\llbracket \boldsymbol{r}\rrbracket_{2d},\llbracket\boldsymbol{r'}\rrbracket_d.$

\label{RP}
\end{algorithm}

\begin{algorithm}
    \caption{Vector-Matrix Mult $\Pi_{VecMatMult}$}
    
    There are $n$ parties where $n=2d+1, 1<k\leq d$ and the number of corrupted parties is $t=d-k+1$. The original vector $\boldsymbol{a}\in\mathbb{F}_p^{u}$ and matrix $A\in\mathbb{F}_p^{u\times v}$.
    
    \KwIn{$\llbracket \boldsymbol{a}\rrbracket_d\in\mathbb{F}_p^{\lceil 
 u/k\rceil},\llbracket A\rrbracket_d\in\mathbb{F}_p^{\lceil u/k\rceil\times v}$.}
    \KwOut{ $\llbracket\boldsymbol{c}\rrbracket_d\in\mathbb{F}_p^{\lceil v/k\rceil},$ where $\boldsymbol{c}=\boldsymbol{a}\cdot A$.}

if $u,v$ is not divisible by $k$, the original vector and matrix can be zero-padded accordingly.
    
    \SetKw{KwCommon}{Common Randomness:}
    \KwCommon{$(\llbracket \boldsymbol{r}\rrbracket_{2d}\in\mathbb{F}_p^{v},\llbracket \boldsymbol{r'}\rrbracket_{d}\in\mathbb{F}_p^{\lceil v/k\rceil})\leftarrow \mathcal{F}_{VM-RandTuple}$} where $\boldsymbol{r}=\{\boldsymbol{r}_i\}^{kv-1}_{i=0},\boldsymbol{r'}=\{\boldsymbol{r}_i=\sum_{j=i\cdot k}^{(i+1)\cdot k -1}\boldsymbol{r}_j\}^{v-1}_{i=0}.$

  \begin{enumerate}[label=\arabic*., align=left, leftmargin=1.2em, labelwidth=1em, labelsep=0.2em]
        \item  The parties compute $\llbracket \boldsymbol{z}\rrbracket_{2d}=\llbracket \boldsymbol{a}\rrbracket_d\cdot\llbracket A\rrbracket_d+\llbracket \boldsymbol{r}\rrbracket_{2d}$ and send $\llbracket \boldsymbol{z}\rrbracket_{2d}$ to the $P_1$.  
        \item $P_1$ reconstructs the values $\{\boldsymbol{z}_i\}_{i=0}^{kv-1}$, computes $\boldsymbol{f}=\{\boldsymbol{f}_i=\sum_{j=i\cdot k}^{(i+1)\cdot k-1}\boldsymbol{z}_j\}^{v-1}_{i=0}$, and distributs $\llbracket\boldsymbol{f}\rrbracket_d\in\mathbb{F}_d^{\lceil v/k\rceil}$ to other parties.
        \item The parties compute $\llbracket\boldsymbol{c}\rrbracket_d=\llbracket\boldsymbol{f}\rrbracket_d-\llbracket\boldsymbol{r'}\rrbracket_d.$
    \end{enumerate}

    \Return $\llbracket\boldsymbol{c}\rrbracket_d.$
    \label{MVMP}
\end{algorithm}





\subsection{Truncation}
\label{trun}
In MPC, decimal numbers are universally encoded using fixed-point representation. In this representation, a decimal value is represented as $\ell$-bit integers, where the first $\ell-\ell_x$ bits are the integer part and the last $\ell_x$ bits are the fractional part. Multiplying two fixed-point numbers yields $2\ell_x$ fractional bits and $\ell-2\ell_x$ integer bits; truncation is required to reduce the fractional part back to $\ell_x$ bits, ensuring precision and avoiding overflow. LXY24 proposes a novel truncation protocol in the Mersenne prime field that reduces the share-to-secret gap to just 1 bit, enabling the use of a 31-bit Mersenne prime. However, this protocol is not well-suited to our setting. First, integrating it into our work would require additional open operations and degree reduction. Second, for deeper neural networks, it necessitates either a 61-bit Mersenne prime or a reduction in precision to avoid overflow. Therefore, we propose a fixed-point vector-matrix multiplication protocol in the 61-bit Mersenne prime field, which serves as a natural extension of ABY3~\cite{Mohassel2018ABY3AM}, striking a balance between communication efficiency and precision in Algorithm~\ref{MVMP_T}. In Fig.~\ref{multtoy}, \ding{172} and \ding{174} correspond to fixed-point vector-matrix multiplication.

 \begin{algorithm}
    \caption{Fixed Vector-Matrix Mult $\Pi_{VecMatMult}^{Fixed}$}
    The original vector $\boldsymbol{a}\in\mathbb{F}_p^{u}$ and matrix $A\in\mathbb{F}_p^{u\times v}$. 
    
    \KwIn{$\llbracket \boldsymbol{a}\rrbracket_d\in\mathbb{F}_p^{\lceil 
 u/k\rceil},\llbracket A\rrbracket_d\in\mathbb{F}_p^{\lceil u/k\rceil\times v}$.}
    \KwOut{$\llbracket\boldsymbol{c}\rrbracket_d\in\mathbb{F}_p^{\lceil v/k\rceil},$ where $\boldsymbol{c}=\boldsymbol{a}\cdot A$.}

if $u,v$ is not divisible by $k$, the original vector and matrix can be zero-padded accordingly.
    
    \SetKw{KwCommon}{Common Randomness:}
    \KwCommon{$(\llbracket \boldsymbol{r}\rrbracket_{2d}\in\mathbb{F}_p^{v},\llbracket \boldsymbol{r'}\rrbracket_{d}\in\mathbb{F}_p^{\lceil v/k\rceil})\leftarrow \mathcal{F}_{TruncTriple}$} where $\boldsymbol{r}=\{\boldsymbol{r}_i\}^{kv-1}_{i=0}, \boldsymbol{r'}=\{\boldsymbol{r}'_i=\sum_{j=i\cdot k}^{(i+1)\cdot k-1}\boldsymbol{r}_j/2^{\ell_x}\}^{v-1}_{i=0}.$

     \begin{enumerate}[label=\arabic*., align=left, leftmargin=1.2em, labelwidth=1em, labelsep=0.2em]
        \item  The parties compute $\llbracket \boldsymbol{z}\rrbracket_{2d}=\llbracket \boldsymbol{a}\rrbracket_d\cdot\llbracket A\rrbracket_d+\llbracket \boldsymbol{r}\rrbracket_{2d}$ and send $\llbracket \boldsymbol{z}\rrbracket_{2d}$ to the $P_1$. 
        \item $P_1$ reconstructs the values $\{\boldsymbol{z}_i\}_{i=0}^{kv-1}$, computes $\boldsymbol{f}=\{\boldsymbol{f}_i=(\sum_{j=i\cdot k}^{(i+1)\cdot k-1}\boldsymbol{z}_j)/2^{\ell_x}\}^{v-1}_{i=0}$ and distributs $\llbracket\boldsymbol{f}\rrbracket_d\in\mathbb{F}_d^{\lceil v/k\rceil}$ to other parties.
        \item The parties compute $\llbracket\boldsymbol{c}\rrbracket_d=\llbracket\boldsymbol{f}\rrbracket_d-\llbracket\boldsymbol{r}'\rrbracket_d.$
    \end{enumerate}
    
    \Return $\llbracket\boldsymbol{c}\rrbracket_d.$
    \label{MVMP_T}
\end{algorithm}

We can see clearly that the key is how to generate the random shares $\llbracket\boldsymbol{r}\rrbracket_{2d},\llbracket\boldsymbol{r}'\rrbracket_{d}$. The vectors $\boldsymbol{r}'$ and $\boldsymbol{r}$ satisfy the following relations: $\boldsymbol{r}'=\{\boldsymbol{r}'_i\}^{k-1}_{i=0},\boldsymbol{r}=\{\boldsymbol{r}_i\}^{k^2-1}_{i=0},\boldsymbol{r}'_i=(\sum_{i=j\cdot k}^{(j+1)\cdot k-1}\boldsymbol{r}_i)/2^{\ell_x},j\in[k]$. We formally define the functionality $\mathcal{F}_{TruncTriple}$ in Appendix~\ref{appendixA} and provide a secure protocol implementation in Algorithm~\ref{TT}. 
Specifically, $\llbracket\boldsymbol{r}'\rrbracket_{d}$ and $k$ random packed secret shares $\llbracket \boldsymbol{w}^{i}\rrbracket_d,i\in[k]$ are genereted, where the length of vector $\boldsymbol{w}^i$ is $k$. Then there is the following relation: $\boldsymbol{r}_j=\boldsymbol{w}_{\lfloor j/k\rfloor}^{j\pmod{k}},j\in[k^2]$. Fig.~\ref{truncation} depicts a sample example. It should be noted that $\llbracket\boldsymbol{w}^i\rrbracket_d$ needs to be recombined. The recombination process consists of two steps: i) All secret values in vector $\boldsymbol{w}^i$ are converted into Shamir secret shares $[\boldsymbol{w}^i_j|_i]_d$. Inspired by~\cite{Zhou2024PrivateNN}, we obtain Theorem~\ref{theoremPSS2SS}. We treat a single packed secret share as $k$ distinct Shamir secret shares and apply Theorem~\ref{theoremPSS2SS} to achieve local conversion. We formalize this process in Algorithm~\ref{PSS2SS}; ii) These transformed values are recombined through linear combination with vector $\llbracket E_i\rrbracket_{k-1}$. Although Algorithm~\ref{PSS2SS} performs the conversion of secret shares locally, it results in an increased degree of $2d$. Therefore, a degree reduction is required before recombination by functionality $\mathcal{F}_{DgreeTrans}$. We provide secure protocol implementation for $\mathcal{F}_{RandomBits},\mathcal{F}_{Random},\mathcal{F}_{DegreeT rans}$ in Appendix~\ref{appendixB}.

\begin{proposition}
If $xy\leq 2^k$ for some $k<\ell$, then with probability at least $1-2^{k}/(2^{\ell}-1)$ it holds that $z = \lfloor xy/2^{\ell_x}\rfloor+v$ for some $v\in \{0, 1\}$.
\end{proposition}
\begin{proof}
    The proof of this proposition is given in Appendix~\ref{appendixC}.
\end{proof}



\begin{proposition}
\label{theoremPSS2SS}
A Shamir secret share $[x|_a]_d,$ where $f(a)=x,n=2d+1,a\notin\{1,2,\cdots,n\},f$ is a degree-$d$ polynomial. Each of the $n$ parties holds its own share $f(i),i\in\{1,2,\cdots, n\}$. The secret value $x$ can be locally converted to $[x|_b]_{2d}$, with the conversion performed as $f'(i)=f(i)\Pi_{j=1,j\neq i}^n\frac{a-j}{b-j}$.
\end{proposition}
\begin{proof}
    The proof of this theorem is given in Appendix~\ref{appendixC}.
\end{proof}

 \begin{algorithm}
    \caption{Trunc Triple $\Pi_{TruncTriple}$}
    \KwIn{None.}
    \KwOut{$\llbracket \boldsymbol{r}\rrbracket_{2d}\in\mathbb{F}_p^{k},\llbracket\boldsymbol{r'}\rrbracket_d\in\mathbb{F}_p$,} 
    
    where $\boldsymbol{r}=\{\boldsymbol{r}_i\}^{k^2-1}_{i=0},$
    
    $\boldsymbol{r}'=\{\boldsymbol{r}'_i=\sum_{j=i\cdot k}^{(i+1)\cdot k -1}\boldsymbol{r}_j/2^{\ell_x}\}^{k-1}_{i=0}.$

     \begin{enumerate}[label=\arabic*., align=left, leftmargin=1.2em, labelwidth=1em, labelsep=0.2em]
        \item  The parties generete $\ell$ random bits PSS $\llbracket\boldsymbol{b}^i\rrbracket_d\leftarrow\mathcal{F}_{RandomBits}$, where the length of each vector $\boldsymbol{b}^i$ is $k,i\in[\ell]$.
        \item Compute $\llbracket\boldsymbol{r'}\rrbracket_{d}=\sum_{i=\ell_x}^{\ell-1}2^{i-\ell_x}\cdot\llbracket\boldsymbol{b}^i\rrbracket_{d},\llbracket\boldsymbol{q}\rrbracket_{d}=\sum_{i=0}^{\ell-1}2^{i}\cdot\llbracket\boldsymbol{b}^i\rrbracket_{d}.$
        \item Generete $k-1$ random PSS $\llbracket \boldsymbol{w}^i\rrbracket_d\leftarrow\mathcal{F}_{Random},i\in[1,k-1]$, and compute $\llbracket \boldsymbol{w}^0\rrbracket_d=\llbracket \boldsymbol{q}\rrbracket_d-\sum_{i=1}^{k-1}\llbracket \boldsymbol{w}^i\rrbracket_d$, where the bit length of each element $\boldsymbol{w}^i_j$ is $\ell$.
        \item Get $[\boldsymbol{w}^i_j|_{s_{i}}]_{2d},i,j\in[k]$ by $\Pi_{ShConvert}(\llbracket \boldsymbol{w}^i\rrbracket_{d})$(Algorithm~\ref{PSS2SS}) locally and get $[\boldsymbol{w}^i_j|_{s_{i}}]_{2d-k+1}\leftarrow\mathcal{F}_{DegreeTrans}([\boldsymbol{w}^i_j|_{s_{i}}]_{2d})$.
        \item Compute $\llbracket\boldsymbol{r}\rrbracket_{2d}=\sum_{i=0}^{k-1}\llbracket E_i\rrbracket_{k-1}\cdot[\boldsymbol{w}^i_j|_{s_{i}}]_{2d-k+1}.$
    \end{enumerate}

    \Return $\llbracket \boldsymbol{r}\rrbracket_{2d},\llbracket\boldsymbol{r'}\rrbracket_d.$
    \label{TT}
\end{algorithm}

\begin{figure}[t]
    \centering
    \includegraphics[width=1\linewidth]{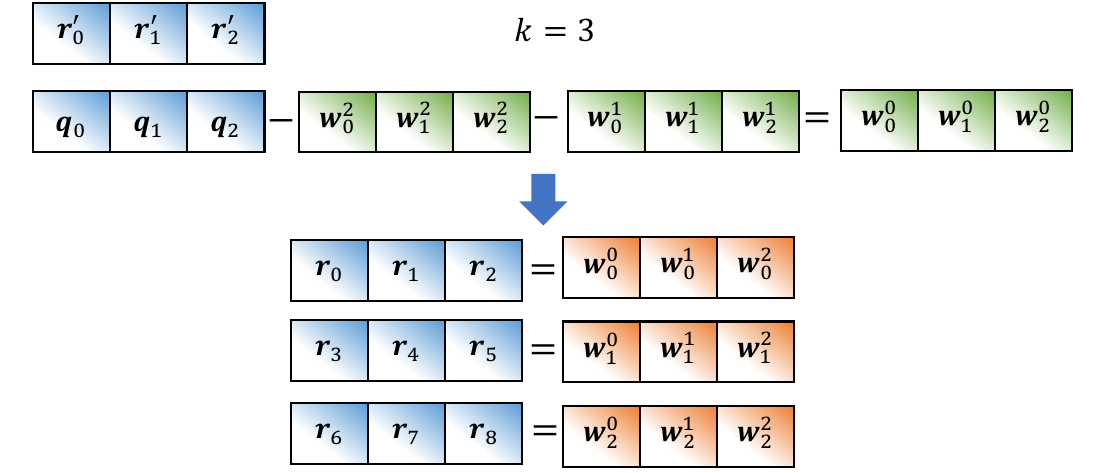}
    \caption{A sample example. $\boldsymbol{q}_i/2^{\ell_x}=\boldsymbol{r}'_i,i\in[3],\boldsymbol{q}_0=(\boldsymbol{r}_0+\boldsymbol{r}_1+\boldsymbol{r}_2),\boldsymbol{q}_1=(\boldsymbol{r}_3+\boldsymbol{r}_4+\boldsymbol{r}_5),\boldsymbol{q}_2=(\boldsymbol{r}_6+\boldsymbol{r}_7+\boldsymbol{r}_8).$}
    \label{truncation}
\end{figure}


 \begin{algorithm}
    \caption{$\Pi_{ShConvert}$}
    \KwIn{$\llbracket \boldsymbol{x}\rrbracket_d,$ where $\llbracket \boldsymbol{x}\rrbracket_d$ can be regarded as $k$ distinct Shamir secret shares at position $\boldsymbol{s}_0$ to $\boldsymbol{s}_{k-1}$.}
    \KwOut{$[\boldsymbol{x}_v|_{\boldsymbol{s}_v}]_{2d},v\in[k]$.}

     \begin{enumerate}[label=\arabic*., align=left, leftmargin=1.2em, labelwidth=1em, labelsep=0.2em]
        \item Each party computes locally \For{$v \leftarrow 0$ \KwTo $k-1$}{
    
    $[\boldsymbol{x}_v|_b]_{2d} \leftarrow [\boldsymbol{x}_v|_{\boldsymbol{s}_v}]_{d}\cdot\Pi_{j=1,j\neq i}^n\frac{\boldsymbol{s}_{v}-j}{b-j},i\in[1,n]$\;
  }
    \end{enumerate}

    \Return $[\boldsymbol{x}_v|_b]_{2d}.$
    \label{PSS2SS}
\end{algorithm}

To sum up, given a vector of length $u$ and a $u$-by-$v$ matrix, the online complexity is the same as protocol $\Pi_{VecMatMult}$. The offline complexity of this protocol depends on the generation of the random bits, random values and degree transformation, that is 5 rounds and $(n-1)(\frac{4\ell+k\ell+2k^2}{n-t}+\frac{4\ell+2k^2}{n}+k\ell+\ell)$ field elements amortized per party for a truncation triple.

\section{Efficient Convolutions}
\label{sec::Conv}
In this section, we first present the design of the convolution computation protocol in \S\ref{PD-DP}. Then, we introduce the connections between layers in \S\ref{PT}.

\subsection{Protocol Design}
\label{PD-DP}
FC is exactly a vector-matrix multiplication. Similarly, a Conv~\S\ref{NN} can also be expressed as a matrix multiplication that can be achieved through multiple vector-matrix multiplications. In light of the above analysis, when zero
padding needs to be applied to the convolutional input, the combination process incurs
significant communication overhead.

To avoid significant communication cost, we present our efficient convolution computation protocol by leveraging the inherent parallelizability of convolution computation. To be specific, the input is convolved with each filter in either a single-input multiple-output or multiple-input multiple-output convolution process as shown in Fig.~\ref{CNNcompute}. In light of this, we propose a filter packing approach, i.e., we pack $k$ values that are selected respectively from the same position of $k$ filters into a single packed secret share. Correspondingly, we also pack $k$ copies of the input into a single packed secret share. Fig.~\ref{CNNp} depicts a sample example. Then we can complete convolution computation by using a parallel matrix multiplication protocol. Further, when zero padding is applied to the input of convolutional operations, it can be performed locally without communication since we can directly pack $k$ zero values into a single packed secret share.

We directly present the protocol for parallel fixed-point matrix multiplication in Algorithm~\ref{PMatMult} and formally define the functionality $\mathcal{F}_{PMatMult}^{Fixed}$ in Appendix~\ref{appendixA}. In offline Step 2, it holds $R_{\alpha,\beta}=\sum_{i=0}^{\ell-1}2^i\cdot R''^i_{\alpha,\beta}\cdot R''^i_{\alpha,\beta},$ since $R''^i_{\alpha,\beta}\in \{0,1\}^k.$ In addition, $\llbracket \boldsymbol{0}\rrbracket_d$ can be generated by PRG without communication. 

To sum up, given a $u$-by-$v$ matrix and a $v$-by-$m$ matrix, the online complexity is 1 round with $\frac{2um}{k}$ field elements amortized per party. The offline complexity is 4 rounds and $\frac{um\ell}{k}(\frac{3n}{n+2k-1}+2+\frac{n+2k}{2n})$ field elements amortized per party.
\begin{algorithm}
    \caption{Fixed Parallel Matrix Mult $\Pi_{PMatMult}^{Fixed}$}
  $2k$ original matrices $A^i\in\mathbb{F}_p^{u\times v},B^i\in\mathbb{F}_p^{v\times m},i\in[k]$. $\llbracket A\rrbracket_d\in\mathbb{F}_p^{u\times v}$ and $\llbracket B\rrbracket_d\in\mathbb{F}_p^{v\times m}$ pack all $A^i$ and all $B^i$, respectively. All $k$ $A^i_{\alpha,\beta},\alpha\in[u],\beta\in[v]$ are packed into a packed secret share $\llbracket A_{\alpha,\beta}\rrbracket_d\in\mathbb{F}_{p}$. The case for all $B^i_{\alpha,\beta},\alpha\in[v],\beta\in[m]$ follow similarly.

    \KwIn{$\llbracket A\rrbracket_d\in\mathbb{F}_p^{u\times v},\llbracket B\rrbracket_d\in\mathbb{F}_p^{v\times m}$.}
     \KwOut{ $\llbracket C\rrbracket_d\in\mathbb{F}_p^{u\times m}$, where $C^i=A^i\cdot B^i$ and $\llbracket C_{\alpha,\beta}\rrbracket_d\in\mathbb{F}_p$ packs $k$ values $\boldsymbol{c}_i,i\in[k]$ and $\boldsymbol{c}_i$ is inner product of the $\alpha$-th row of $A^i$ and the $\beta$-th column of $B^i$.}

    \SetKw{KwOffline}{Offline:}
    \KwOffline{}

     \begin{enumerate}[label=\arabic*., align=left, leftmargin=1.2em, labelwidth=1em, labelsep=0.2em]
        \item  Generate $u\cdot m\cdot\ell$ random bits PSS $\llbracket R''^i_{\alpha,\beta} \rrbracket_d\leftarrow\mathcal{F}_{RandomBits},\alpha\in[u],\beta\in[m],i\in[\ell]$, where the length of each $R''^i_{\alpha,\beta}$ is $k$.
        \item Get $\llbracket R_{\alpha,\beta}\rrbracket_{2d}=\sum_{i=0}^{\ell-1}2^{i}\cdot\llbracket R''^i_{\alpha,\beta}\rrbracket_d\cdot\llbracket R''^i_{\alpha,\beta}\rrbracket_d+\llbracket \boldsymbol{0}\rrbracket_d.$
        \item Get $\llbracket R'_{\alpha,\beta}\rrbracket_{d}=\sum_{i=\ell-\ell_x}^{\ell-1}2^{i-\ell_x}\cdot\llbracket R''^i_{\alpha,\beta}\rrbracket_d.$
    
    \end{enumerate}

    \SetKw{KwOnline}{Online:}
    \KwOnline{}
         \begin{enumerate}[label=\arabic*., align=left, leftmargin=1.2em, labelwidth=1em, labelsep=0.2em]
        \item Compute $\llbracket C\rrbracket_{2d}=\llbracket A \rrbracket_d\cdot \llbracket B \rrbracket_d+\llbracket R\rrbracket_d$ and send $\llbracket C\rrbracket_{2d}$ to the $P_1$.
        \item $P_1$ reconstructs the matrix $C$, computes $C'=C/2^{\ell_x}$ and distributes $\llbracket C'\rrbracket_d$ to the other parties.
        \item  Compute $\llbracket C\rrbracket_d=\llbracket C'\rrbracket_d-\llbracket R'\rrbracket_d.$
    \end{enumerate}
    \Return $\llbracket C\rrbracket_d.$
    \label{PMatMult}
\end{algorithm}


\subsection{Inter-layer Connections}
\label{PT}
The number of output channels in a convolutional layer is determined by the number of filters. Each filter produces an independent output channel. After accomplishing the convolution computation by using our parallel matrix multiplication protocol, the $k$ values that are selected respectively from the same position of $k$ output channels are packed into a single packed secret share as shown in Fig.~\ref{CNNp}. It should be noted that the results may require further processing before being propagated to the next layer. 

When the next layer performs a non-linear operation, the result can be directly utilized as input since non-linear layer operations are performed element-wise~\S\ref{sec::nonline}. 

When the next layer is an FC layer, seamless integration can be achieved without requiring any communication. The reason is that the results of the convolution computation require flattening to a vector and the flattened vector may be regarded as already packed through block-wise packing approach in~\S\ref{sec::Vector-Matrix Multiplication}. Correspondingly, we only need to pack the weight matrix according to the index of each element of the flattened vector. 

When the next layer still performs a Conv operation, a pack transformation is required for the result of the convolution computation since we need to pack $k$ copies of the input into a single packed secret share. To be specific, the pack transformation is that given a PSS $\llbracket\boldsymbol{x}\rrbracket_d$, it outputs $\llbracket\boldsymbol{x}^i\rrbracket_d$, where $\boldsymbol{x}$ and $\boldsymbol{x}^i$ are two vectors of length $k,\boldsymbol{x}^i=(\boldsymbol{x}_i,\cdots,\boldsymbol{x}_i),i\in[k]$. Fig.~\ref{CNNp} presents an example. We formally define the functionality $\mathcal{F}_{PackTrans}$ in Appendix~\ref{appendixA} and provide a secure protocol implementation in Algorithm~\ref{PackTrans}. In the offline, we generate $\llbracket \boldsymbol{r}'^i\rrbracket_d,\llbracket \boldsymbol{r}\rrbracket_d,$ where $\llbracket\boldsymbol{r}\rrbracket_d=(\boldsymbol{r}_0,\cdots,\boldsymbol{r}_{k-1}),\llbracket\boldsymbol{r}'^i\rrbracket_d=(\boldsymbol{r}_i,\cdots,\boldsymbol{r}_{i}),i\in[k].$ In the online, parties compute $\llbracket\boldsymbol{z}\rrbracket_d=\llbracket\boldsymbol{x}\rrbracket_d+\llbracket\boldsymbol{r}\rrbracket_d$ and reveal the vector $\boldsymbol{z}$ to $P_1$. Then, $P_1$ sets $\boldsymbol{z}^i=(\boldsymbol{z}_i,\cdots,\boldsymbol{z}_i)$ and distributes $\llbracket\boldsymbol{z}^i\rrbracket_d$. Finally, parties compute $\llbracket \boldsymbol{x}^i \rrbracket_d=\llbracket\boldsymbol{z}^i\rrbracket_d-\llbracket\boldsymbol{r}'^i\rrbracket_d$.

To sum up, the online complexity is 1 round and $(n-1)(\frac{1}{2}+k)$ field elements amortized per party. The offline complexity is 2 rounds and  $(n-1)(\frac{2+k}{n-t}+\frac{2}{n})$ field elements amortized per party.
\begin{figure*}
    \centering
    \includegraphics[width=1\linewidth]{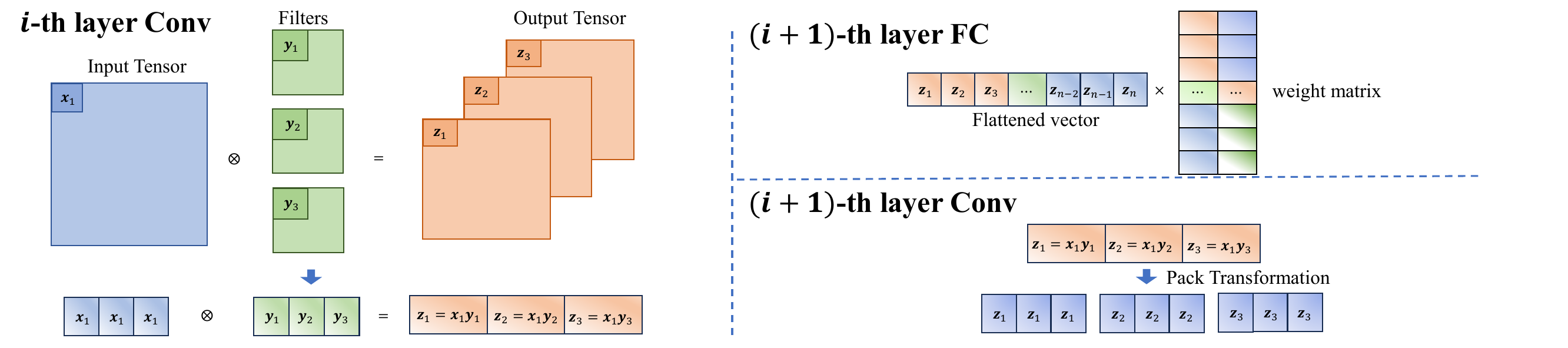}
    \caption{A toy example where $k=3$. Each pixel of the input tensor needs to be multiplied by each filter. $\boldsymbol{y}_1,\boldsymbol{y}_2,\boldsymbol{y}_3$ are packed into a single packed secret share. Three copies of $\boldsymbol{x}_1$ are packed into a single packed secret share. If the next layer is still a Conv operation, the result $\llbracket \boldsymbol{z}\rrbracket_{d}$ needs to be converted to $\llbracket \boldsymbol{z}^i\rrbracket_{d},i\in\{1,2,3\},\boldsymbol{z}=(\boldsymbol{z}_1,\boldsymbol{z}_2,\boldsymbol{z}_3),\boldsymbol{z}^1=(\boldsymbol{z}_1,\boldsymbol{z}_1,\boldsymbol{z}_1),\boldsymbol{z}^2=(\boldsymbol{z}_2,\boldsymbol{z}_2,\boldsymbol{z}_2),\boldsymbol{z}^3=(\boldsymbol{z}_3,\boldsymbol{z}_3,\boldsymbol{z}_3)$ by pack transformation protocol. If the next layer is an FC operation, the result $\llbracket \boldsymbol{z}\rrbracket_{d}$ can become a flattened vector seamlessly and the weight matrix can be packed according to the index of each element
of the flattened vector.}
    \label{CNNp}
\end{figure*}

\begin{algorithm}
    \caption{Pack Transformation $\Pi_{PackTrans}$}
    \KwIn{$\llbracket \boldsymbol{x}\rrbracket_d\in\mathbb{F}_p$.}
    \KwOut{ $\llbracket \boldsymbol{x}^i\rrbracket_d\in\mathbb{F}_p,i\in[k]$, where $\boldsymbol{x}^i$ is a vector of length $k$ and each element is equivalent to $\boldsymbol{x}_i$.}

    \SetKw{KwOffline}{Offline:}
    \KwOffline{}

     \begin{enumerate}[label=\arabic*., align=left, leftmargin=1.2em, labelwidth=1em, labelsep=0.2em]
        \item Generare $\llbracket \boldsymbol{r}'^i\rrbracket_d\leftarrow\mathcal{F}_{Random}$, where $\llbracket\boldsymbol{r}'^i\rrbracket_d=(\boldsymbol{r}_i,\cdots,\boldsymbol{r}_{i}),\boldsymbol{r}_i$ is a random value.
        \item  Compute $\llbracket \boldsymbol{r} \rrbracket_{d+k-1}=\sum_{i=0}^{k-1}\llbracket E_{i}\rrbracket_{k-1}\cdot\llbracket \boldsymbol{r}'^i\rrbracket_d,$ where $\boldsymbol{r}=(\boldsymbol{r}_0,\cdots,\boldsymbol{r}_{k-1}).$
        \item Get $\llbracket \boldsymbol{r} \rrbracket_d\leftarrow \mathcal{F}_{DegreeTrans}(\llbracket \boldsymbol{r} \rrbracket_{d+k-1}).$
    \end{enumerate}

    \SetKw{KwOnline}{Online:}
    \KwOnline{}
    
     \begin{enumerate}[label=\arabic*., align=left, leftmargin=1.2em, labelwidth=1em, labelsep=0.2em]
        \item Compute $\llbracket \boldsymbol{z}\rrbracket_{d}=\llbracket \boldsymbol{x}\rrbracket_d+\llbracket \boldsymbol{r}\rrbracket_d$ and send $\llbracket \boldsymbol{z}\rrbracket_{d}$ to the $P_1$.
        \item $P_1$ reconstructs the vector $\boldsymbol{z}=(\boldsymbol{z}_0,\cdots,\boldsymbol{z}_{k-1})$, sets $\boldsymbol{z}^i=(\boldsymbol{z}_i,\boldsymbol{z}_i,\cdots,\boldsymbol{z}_i)$ and distributes $\llbracket \boldsymbol{z}^i\rrbracket_d$ to the other parties.
        \item Compute $\llbracket \boldsymbol{x}^i\rrbracket_d=\llbracket \boldsymbol{z}^i\rrbracket_d-\llbracket \boldsymbol{r'}^i\rrbracket_d,i\in[k].$
    \end{enumerate}
    \Return $\llbracket \boldsymbol{x}^i\rrbracket_d.$
    \label{PackTrans}
\end{algorithm}



\section{Non-Linear Layers}
\label{sec::nonline}
In this section, we introduce our efficient protocols for non-linear operations such as parallel prefix-ORs and bitwise less than in \S\ref{Prefix-ORs-subsection}, parallel derivative of ReLU (DReLU) and ReLU in \S\ref{DReLU-subsection}, and parallel Maxpool in \S\ref{Maxpool-subsection}.

\subsection{Parallel Prefix-ORs and Bitwise Less Than}
\label{Prefix-ORs-subsection}
We first describe our method to compute the prefix-ORs, which is the core and bottleneck to realizing the bitwise less-than functionality. The prefix-ORs operation ($b_j=\lor_{i=0}^{j}a_i,a_i\in\{0,1\},j\in[\ell]$) can be implemented via its dual problem, that is, $\bar{b}_j=\land_{i=0}^{j}\bar{a}_i.$ It is worth noting that LXY24 designs a protocol for prefix-ORs with only one round of communication; however,~\cite{zhang2025helix} identified a privacy leakage issue in that protocol. In this work, we avoid this privacy leakage by utilizing functionality $\mathcal{F}_{PDN}$ and provide a secure protocol implementation for parallel prefix-ORs in Algorithm~\ref{Preor}. $\mathcal{F}_{PreMult}$ denotes parallel prefix multiplication functionality and we formally define the functionality in Appendix~\ref{appendixA}. Further, $\mathcal{F}_{PreMult}$ can be securely implemented by directly employing $\mathcal{F}_{PDN}$ as a subroutine. As illustrated in Fig.~\ref{premult}, an 8-input prefix multiplication is depicted, where \ieeediamond{} denotes a multiplication gate.

To sum up, the online communication complexity is $\lceil \text{log}_2(\ell)\rceil$ rounds with $\frac{\lceil\ell \text{log}_2(\ell)\rceil}{k}$ field elements amortized per party. The offline communication complexity is 1 rounds with $\frac{n\ell}{k(n-t)}\lceil \text{log}_2(\ell)\rceil$ field elements amortized per party.

\begin{algorithm}[t]
    \caption{Prefix-ORs $\Pi_{PreOR}$}
    \KwIn{$(\llbracket \boldsymbol{a}^0\rrbracket_d,\cdots,\llbracket \boldsymbol{a}^{\ell-1}\rrbracket_d)$, where $\boldsymbol{a}^i$ is a bit vector of length $k,\llbracket \boldsymbol{a}^i\rrbracket_d\in\mathbb{F}_p,\boldsymbol{a}^i=(\boldsymbol{a}^i_0,\cdots,\boldsymbol{a}^i_{k-1})$.}
    \KwOut{ All prefix ORs $\llbracket\boldsymbol{a}^0\rrbracket_d,\llbracket\boldsymbol{a}^0\lor\boldsymbol{a}^1\rrbracket_d,\cdots, \llbracket\lor_{i=0}^{\ell-1}\boldsymbol{a}^i\rrbracket_d$, where $\lor_{i=0}^{\ell-1}\boldsymbol{a}^i=(\lor_{i=0}^{\ell-1}\boldsymbol{a}^i_0,\cdots,\lor_{i=0}^{\ell-1}\boldsymbol{a}^i_{k-1}).$}

     \begin{enumerate}[label=\arabic*., align=left, leftmargin=1.2em, labelwidth=1em, labelsep=0.2em]
        \item  Compute $\llbracket \boldsymbol{b}^i\rrbracket_d=1-\llbracket\boldsymbol{a}^i\rrbracket_d$.
        \item  Get $\llbracket \boldsymbol{c}^i\rrbracket_d\leftarrow\mathcal{F}_{PreMult}(\llbracket \boldsymbol{b}^0\rrbracket_d,\cdots,\llbracket \boldsymbol{b}^{\ell-1}\rrbracket_d)$.
        \item Compute $1-\llbracket \boldsymbol{c}^0\rrbracket_d,\cdots,1-\llbracket \boldsymbol{c}^{\ell-1}\rrbracket_d$.
    \end{enumerate}
    
    \Return $\llbracket\boldsymbol{a}^0\rrbracket_d,\llbracket\boldsymbol{a}^0\lor\boldsymbol{a}^1\rrbracket_d,\cdots, \llbracket\lor_{i=0}^{\ell-1}\boldsymbol{a}^i\rrbracket_d.$
    
    \label{Preor}
\end{algorithm}


Based on PSS, multiple bitwise less-than circuits can be securely computed in parallel. Specifically, given $k$ public values $a_0,\cdots,a_{k-1}$ and $\ell$ shared bits vectors $\llbracket\boldsymbol{b}^i\rrbracket_d,i\in[\ell],\boldsymbol{b}^i\in\{0,\cdots,1\}^k$, it outputs the share of $(a_j<b_j),b_j=\sum_{i=0}^{\ell-1}2^i\boldsymbol{b}^i_j,j\in[k]$. Our protocol $\Pi_{Bitwise-LT}$ follows the
 method in LXY24, with the substitution of our secure prefix-ORs subprotocol $\Pi_{PreOR}$. Algorithm~\ref{BitLT} depicts the protocol $\Pi_{Bitwise-LT}$. In Step 4, the XOR operation can be computed by $\llbracket \boldsymbol{a}^i \rrbracket_{k-1}\oplus\llbracket \boldsymbol{b}^i \rrbracket_{d}=\llbracket \boldsymbol{a}^i \rrbracket_{k-1}+\llbracket \boldsymbol{b}^i \rrbracket_{d}-2\cdot\llbracket \boldsymbol{a}^i \rrbracket_{k-1}\cdot\llbracket \boldsymbol{b}^i \rrbracket_{d}.$ We formally define the functionality $\mathcal{F}_{Xor}$ in Appendix~\ref{appendixA} and provide a secure protocol implementation in Appendix~\ref{appendixB}.
 
To sum up, the online communication complexity is $\lceil \text{log}_2(\ell)\rceil+2$ 
rounds with $\frac{2\ell+2+\lceil\ell \text{log}_2(\ell)\rceil}{k}$ field elements amortized per party. The offline
communication complexity is 1 rounds with $\frac{n}{k(n-t)}(4+\lceil \ell \text{log}_2(\ell)\rceil)$ field elements amortized per party.


\begin{figure}
    \centering
    \includegraphics[width=0.8\linewidth]{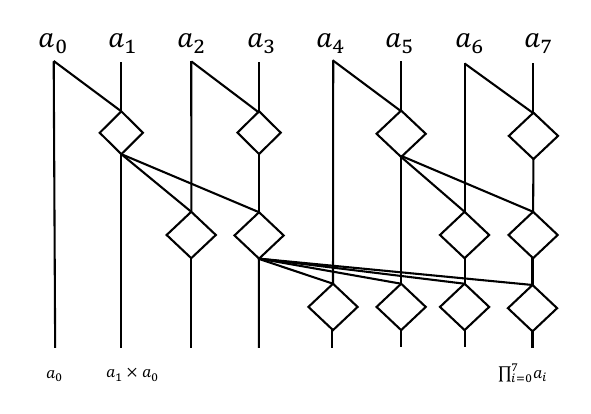}
    \caption{An 8-input prefix multiplication. \protect\ieeediamond{}  denotes a multiplication gate that can be securely implemented by $\mathcal{F}_{PDN}$.}
    \label{premult}
\end{figure}

\begin{algorithm}
    \caption{Bitwise Less Than $\Pi_{Bitwise-LT}$}
    \KwIn{$k$ public values $a_0,\cdots,a_{k-1}$, $\ell$ shared bits vectors $\llbracket\boldsymbol{b}^i\rrbracket_d,i\in[\ell],\boldsymbol{b}^i\in\{0,\cdots,1\}^k,b_j=\sum_{i=0}^{\ell-1}2^i\boldsymbol{b}^i_j$.}
    \KwOut{$\llbracket a_j<b_j \rrbracket_d,j\in[k]$}

    \begin{enumerate}[label=\arabic*., align=left, leftmargin=1.2em, labelwidth=1em, labelsep=0.2em]
        \item   Let $a_{0,0},\cdots,a_{0,\ell-1},\cdots,a_{k-1,0},\cdots,a_{k-1,\ell-1}$ be the decomposed bits of $a_0,\cdots,a_{k-1}$.
        \item Set $\bar{a}_{j,i}=1-a_{j,i}$ and $\llbracket\bar{\boldsymbol{b}}^i\rrbracket_d=1-\llbracket\boldsymbol{b}^i\rrbracket_d,i\in[\ell],j\in[k].$
        \item Compute $\llbracket \boldsymbol{a}^i \rrbracket_{k-1}$, where $\boldsymbol{a}^i=(\bar{a}_{0,i},\cdots,\bar{a}_{k,i}),i\in[\ell]$.
        \item Compute $\llbracket \boldsymbol{c}^i \rrbracket_{d}\leftarrow\mathcal{F}_{Xor}(\llbracket \boldsymbol{a}^i \rrbracket_{k-1},\llbracket \bar{\boldsymbol{b}}^i \rrbracket_{d})$.
        \item Get $\llbracket \boldsymbol{f}^{\ell-1} \rrbracket_{d},\cdots,\llbracket \boldsymbol{f}^0 \rrbracket_{d}\leftarrow\mathcal{F}_{PreOR}(\llbracket \boldsymbol{c}^{\ell-1} \rrbracket_{d},\cdots,\llbracket \boldsymbol{c}^0 \rrbracket_{d})$.
        \item Set $\llbracket \boldsymbol{h}^{i} \rrbracket_{d}=\llbracket \boldsymbol{f}^i \rrbracket_{d}-\llbracket \boldsymbol{f}^{i+1} \rrbracket_{d}$, where $\llbracket \boldsymbol{h}^{\ell-1} \rrbracket_{d}=\llbracket \boldsymbol{f}^{\ell-1} \rrbracket_{d}.$
        \item Compute $\llbracket a_j<b_j \rrbracket_{d+k-1}=\sum_{i=0}^{\ell-1}\llbracket \boldsymbol{a}^i\rrbracket_{k-1}\cdot\llbracket \boldsymbol{h}^i\rrbracket_{d}.$
        \item Output $\llbracket a_j<b_j \rrbracket_{d}\leftarrow \mathcal{F}_{DegreeTrans}(\llbracket a_j<b_j \rrbracket_{d+k-1}),j\in[k].$
    \end{enumerate}

    \Return  $\llbracket a_j<b_j \rrbracket_d,j\in[k]$.
    
    \label{BitLT}
\end{algorithm}

\subsection{Parallel DReLU and ReLU}
\label{DReLU-subsection}
 ReLU is a common activation function in neural networks. It can be represented as $\text{ReLU}(x) =\text{Max}(0, x)$ and the derivative of ReLU (DReLU) is 1 for $x\geq0$ and 0 for $x<0$.
\begin{center}
    DReLU$(x)= \left\{ \begin{array}{ll}
1,&x \geq 0 \\
0,&x < 0
\end{array} \right.$
\end{center}
Since integers are represented in binary two’s complement form, the most significant bit of an integer is related to DReLU. Furthermore, Mersenne primes are necessarily odd, which implies that $\text{MSB}(a)=\text{LSB}(2a)$ according to~\cite{Wagh2019SecureNN3S}. Therefore, it follows that $\text{DReLU}(a)=1-\text{LSB}(2a)$. As in~\cite{Wagh2019SecureNN3S,Liu2024ScalableMC}, in order to compute securely $\text{LSB}(2a)$, we mask $2a$ with a random value $r$ such that $y=2a+r$ and reveal $y$. Consequently, $\text{LSB}(y)=\text{LSB}(2a)\oplus\text{LSB}(r)$ if $y\geq r$, and $\text{LSB}(y)=\text{LSB}(2a)\oplus\text{LSB}(r)\oplus1$ otherwise and we can get $\text{LSB}(2a)=\text{LSB}(y)\oplus\text{LSB}(r)\oplus(y<r)$.

We employ the PSS to enable parallel computation of multiple DReLU.  Specifically, given $k$ public values $a_0,\cdots,a_{k-1}$, it outputs the of $\text{DReLU}(a_i),i\in[k]$. We provide a secure protocol implementation in Algorithm~\ref{DReLU}. 

To sum up, the online communication complexity is $\lceil \text{log}_2(\ell)\rceil+5$ 
rounds with $\frac{\ell\mathrm{log}(\ell)+2\ell+7+k}{k}$ field elements amortized per party. The offline
communication complexity is 4 rounds with $\frac{n}{k}(\frac{1}{n-t}(13+\lceil \ell \text{log}_2(\ell)\rceil)+\frac{4}{n}+k+1)$ field elements amortized per party.

Note that $\text{ReLU}(a)=\text{DReLU}(a)\cdot a$. So we can easily implement parallel computation of multiple $\text{ReLU}$ functions. The secure protocol is in Algorithm~\ref{ReLU}.  

To sum up, the online communication complexity is $\lceil \text{log}_2(\ell)\rceil+6$ 
rounds with $\frac{\ell\mathrm{log}(\ell)+2\ell+9+k}{k}$ field elements amortized per party. The offline
communication complexity is 4 rounds with $\frac{n}{k}(\frac{1}{n-t}(15+\lceil \ell \text{log}_2(\ell)\rceil)+\frac{4}{n}+k+1)$ field elements amortized per party.

\begin{algorithm}
    \caption{DReLU $\Pi_{DReLU}$}
    \KwIn{$\llbracket\boldsymbol{a}\rrbracket_d,\boldsymbol{a}$ is a vector with length $k$.}
    \KwOut{$\llbracket \text{DReLU}(\boldsymbol{a}) \rrbracket_d.$}

    \SetKw{KwOffline}{Offline:}
    \KwOffline{}

    \begin{enumerate}[label=\arabic*., align=left, leftmargin=1.2em, labelwidth=1em, labelsep=0.2em]
        \item  Generate $\ell$ packed bit secret shares $\llbracket \boldsymbol{r}^i\rrbracket_d\leftarrow\mathcal{F}_{RandomBits}, i\in[\ell]$, where $\boldsymbol{r}^i$ is a bit vector with length $k$.
        \item Compute $\llbracket \boldsymbol{r}\rrbracket_d=\sum_{i=0}^{\ell-1}2^i\cdot\llbracket \boldsymbol{r}^i\rrbracket_d,$ where $\boldsymbol{r}_j=\sum_{i=0}^{\ell-1}2^i\cdot\boldsymbol{r}^i_j,j\in[k].$
     
    \end{enumerate}

    \SetKw{KwOnline}{Online:}
    \KwOnline{}

    \begin{enumerate}[label=\arabic*., align=left, leftmargin=1.2em, labelwidth=1em, labelsep=0.2em]
        \item  Compute $\llbracket\boldsymbol{y}\rrbracket_d=2\llbracket\boldsymbol{a}\rrbracket_d+\llbracket\boldsymbol{r}\rrbracket_d$, send $\llbracket\boldsymbol{y}\rrbracket_d$ to $P_1$ and $P_1$ reconstructs and opens $\boldsymbol{y}=(\boldsymbol{y_0},\cdots,\boldsymbol{y}_{k-1}).$
        \item Compute Let $\boldsymbol{y}_{0,0},\cdots,\boldsymbol{y}_{0,\ell-1}, \cdots, \boldsymbol{y}_{k-1,0}, \cdots, \boldsymbol{y}_{k-1,\ell-1}$ be the decomposed bits of $\boldsymbol{y}_0, \cdots, \boldsymbol{y}_{k-1}$.
        \item Compute $\llbracket \boldsymbol{y}^0 \rrbracket_{k-1}$, where $\boldsymbol{y}^0=(\boldsymbol{y}_{0,0}, \cdots, \boldsymbol{y}_{k,0})$.
        \item Compute $\llbracket \boldsymbol{b}\rrbracket_{d}\leftarrow\mathcal{F}_{Xor}(\llbracket \boldsymbol{y}^0 \rrbracket_{k-1},\llbracket \boldsymbol{r}^0 \rrbracket_{d})$.
        \item Get $\llbracket \boldsymbol{c}\rrbracket_{d}\leftarrow\mathcal{F}_{Bitwise-LT}(\boldsymbol{y}_0, \cdots, \boldsymbol{y}_{k-1},\llbracket \boldsymbol{r}^i \rrbracket_{d}),i\in[\ell]$.
        \item Compute $\llbracket \boldsymbol{f}\rrbracket_{d}\leftarrow\mathcal{F}_{Xor}(\llbracket \boldsymbol{b} \rrbracket_{d},\llbracket \boldsymbol{c} \rrbracket_{d})$.
        \item Output $1-\llbracket \boldsymbol{f}\rrbracket_{d}.$
     
    \end{enumerate}
    \Return  $\llbracket \text{DReLU}(\boldsymbol{a}) \rrbracket_d.$
    
    \label{DReLU}
\end{algorithm}


\begin{algorithm}
    \caption{ReLU $\Pi_{ReLU}$}
    \KwIn{$\llbracket\boldsymbol{a}\rrbracket_d,\boldsymbol{a}$ is a vector with length $k$.}
    \KwOut{$\llbracket \text{ReLU}(\boldsymbol{a}) \rrbracket_d.$}
    \begin{enumerate}[label=\arabic*., align=left, leftmargin=1.2em, labelwidth=1em, labelsep=0.2em]
        \item Get $\llbracket\boldsymbol{b}\rrbracket_d\leftarrow\mathcal{F}_{DReLU}(\boldsymbol{a}).$
        \item Output $\llbracket\text{ReLU}(\boldsymbol{a})\rrbracket_d\leftarrow\mathcal{F}_{PDN}(\llbracket\boldsymbol{a}\rrbracket_d,\llbracket\boldsymbol{b}\rrbracket_d).$
    \end{enumerate}
    
    \Return  $\llbracket \text{ReLU}(\boldsymbol{a}) \rrbracket_d.$
    
    \label{ReLU}
\end{algorithm}


\subsection{Parallel Maxpool}
\label{Maxpool-subsection}
 As described in \S\ref{NN}, Maxpool is the key component in convolutional neural networks (CNN). The Maxpool can be reformulated as finding the maximum value in a vector of length $m$. We can represent $\text{Max}(a,b)$ as follows:
\begin{center}
    $\text{Max}(a,b)=\text{ReLU}(a-b)+b$
\end{center}
So Maxpool can be realized through ReLU. Assume that $m$ is a power of two. To identify the maximum among $m$ elements, the values are first grouped into pairs, and a comparison is performed within each pair to determine the local maximum, yielding a new vector of size $m/2$. This procedure is applied recursively over $\text{log}_2(m)$ iterations until a single maximum value is obtained. 

Similarly, we use PSS to enable the parallel computation of multiple Maxpool operations. Specifically, given $m$ shared vectors $\llbracket \boldsymbol{a}^i\rrbracket_d,$ where $\boldsymbol{a}^i$ is vector with length $k,i\in[m]$. It outputs the share of $\text{Max}(\boldsymbol{a}^1_j,\cdots,\boldsymbol{a}^m_{j}),j\in[k]$.
We provide a secure protocol $\Pi_{Maxpool}$ implementation in Algorithm~\ref{Maxpool}. 

To sum up, the online communication complexity is $\lceil \text{log}_2(m)\rceil(\lceil \text{log}_2(\ell)\rceil+6)$ 
rounds with $\frac{(m-1)(\ell\mathrm{log}(\ell)+2\ell+9+k)}{k}$ field elements amortized per party. The offline
communication complexity is 4 rounds with $\frac{n}{k}(\frac{1}{n-t}(15+\lceil \ell \text{log}_2(\ell)\rceil)+\frac{4}{n}+k+1)(m-1)$ field elements amortized  per party.

\begin{algorithm}[t]
    \caption{Maxpool $\Pi_{Maxpool}$}
    \KwIn{$\llbracket\boldsymbol{a}^i\rrbracket_d,\boldsymbol{a}$ is a vector with length $k,i\in[m]$. Assuming $m$ is a power of two.}
    \KwOut{$\llbracket\boldsymbol{b}\rrbracket_d,$ where $\boldsymbol{b}$ is a vector with length $k$ and $\boldsymbol{b}_j=\text{Max}(\boldsymbol{a}^1_j,\cdots,\boldsymbol{a}^m_j),j\in[k].$}

        \begin{enumerate}[label=\arabic*., align=left, leftmargin=1.2em, labelwidth=1em, labelsep=0.2em]
        \item If $m=1:$ Output $\llbracket\boldsymbol{a}^0\rrbracket_d$.
        \item Set $s=m$ and $\boldsymbol{c^1}=(\llbracket\boldsymbol{a}^0\rrbracket_d,\cdots,\llbracket\boldsymbol{a}^{m-1}\rrbracket_d).$
        \item \For{ $i\leftarrow 0$ \KwTo $\text{log}_2(m)-1$}{
            Set $s=s/2$,
            $(\llbracket \boldsymbol{f}^0\rrbracket_d,\cdots,\llbracket\boldsymbol{f}^{s-1} \rrbracket_d,\llbracket \boldsymbol{h}^0\rrbracket_d,\cdots,\llbracket\boldsymbol{h}^{s-1} \rrbracket_d)=\boldsymbol{c}^i.$

            \For{ $j\leftarrow 0$ \KwTo $s-1$}{
         $\llbracket \boldsymbol{g}^j\rrbracket_d\leftarrow\mathcal{F}_{ReLU}(\llbracket\boldsymbol{f}^j-\boldsymbol{h}^j\rrbracket_d)+\boldsymbol{h}^j$  
        }
        If $s=1:$ Output $\llbracket \boldsymbol{g}^0\rrbracket_d$.
        
        Else: Set $\boldsymbol{c}^{i+1}\leftarrow(\llbracket\boldsymbol{g}^{0}\rrbracket_d,\cdots,\llbracket\boldsymbol{g}^{s-1}\rrbracket_d)$.
        }

    \end{enumerate}

    \Return  $\llbracket\boldsymbol{b}\rrbracket_d.$
    
    \label{Maxpool}
\end{algorithm}


\section{Evaluation}
\subsection{Experiments Environment}
We implement this framework on top of LXY24 in C++ and run our experiments on 11 Alibaba Cloud servers, each of which was equipped with an AMD EPYC 32-core processor with 64GB of RAM. We use the Eigen library to accelerate matrix multiplication and parallelize the computations using multithreading. To simulate different network conditions, we use Linux Traffic Control (tc) to simulate a wide-area network (WAN, RTT: 40 ms, 100 Mbps), and a local-area network (LAN, RTT: 0.2 ms, 10 Gbps). RTT denotes Round-Trip Time. We completed performance evaluations for inference in various settings involving different numbers of parties, including 11-party computation (11PC), 21PC, 31PC, and 63PC. In addition, we rerun the protocol of LXY24 under the same adversarial model. Since the single-round communication prefix multiplication protocol in LXY24 suffers from privacy leakage issues, we replace it with a logarithmic-round prefix multiplication protocol based on the DN multiplication protocol. In the experiment, all arithmetic operations on the shares are performed modulo $p = 2^{61}-1$, with a fixed-point precision of 13 bits. All evaluations are executed 10 times, and the average is recorded. 

\begin{table*}[ht!]
\centering
\caption{Online round and communication complexity of the main protocols. Numbers of the amortized communication are reported in field elements per party. The shapes of the fixed matrix multiplication with truncation are $u \times v$ and $v \times o$. 2D convolution $(\mathrm{Conv}_{m,c,o,f,s})$ with the inputs of size $m \times m$, $c$ input channels, $o$ kernels of size $f \times f$ , and the stride $s$. InLXY24, the input to the prefix-OR $(\mathrm{PreOR}_{\ell})$ is a single vector of length $\ell$. In ours, the input consists of $k$ vectors, each of length $\ell$. The input to several building blocks, including bitwise less than $(\mathrm{BitwiseLT}_{k})$, $\mathrm{DReLU}_{k}$, $\mathrm{ReLU}_{k}$ and $\mathrm{Maxpool}_{k}$, is a vector of length $k$.}
\label{table::TROur}
\renewcommand{\arraystretch}{1.2}  
\begin{tabular}{|c|cc|cc|}
\hline
\multirow{2}{*}{Protocol} & \multicolumn{2}{c|}{Rounds} & \multicolumn{2}{c|}{Communication}\\
\cline{2-5}
&LXY24 &ours &LXY24  & ours \\
\hline
$\mathrm{MatMult}_{u,v,o}(\Pi_{VecMatMult}^{Fixed})$& 1 & 1& $2uv$ & $(k+1)uv/k$  \\
\hline
$\mathrm{Conv}_{m,c,o, f,s}(\Pi_{PMatMult}^{Fixed})$& 1 & 1 & $2((m-f)/s + 1)^2o$ & $2((m-f)/s + 1)^2o/k$  \\
\hline
$\mathrm{PreOR}_{\ell}(\Pi_{PreOR})$& $\mathrm{log}(\ell)$ & $\mathrm{log}(\ell)$ & $\ell\mathrm{log}(\ell)$ & $\ell\mathrm{log}(\ell)/k$  \\
\hline
$\mathrm{BitwiseLT}_{k}(\Pi_{Bitwise-LT})$& $\mathrm{log}(\ell)$ & $log(\ell)+2$ & $k\ell\mathrm{log}(\ell)$& $\ell\mathrm{log}(\ell)+2\ell+2$  \\
\hline
$\mathrm{DReLU}_{k}(\Pi_{DReLU})$& $\mathrm{log}(\ell)+1$ &$\mathrm{log}(\ell)+ 5$ & $k(\ell\mathrm{log}(\ell)+4)$ & $\ell\mathrm{log}(\ell)+2\ell+7+k$  \\
\hline
$\mathrm{ReLU}_{k}(\Pi_{ReLU})$& $\mathrm{log}(\ell)+2$ & $\mathrm{log}(\ell)+6$ & $k(\ell\mathrm{log}(\ell)+ 6)$ & $\ell\mathrm{log}(\ell)+2\ell+9+k$  \\
\hline
$\mathrm{Maxpool}_{k}(\Pi_{Maxpool})$& $\mathrm{log}(k)(\mathrm{log}(\ell)+2)$ &$\mathrm{log}(k)(\mathrm{log}(\ell)+6)$ & $(k-1)(\ell\mathrm{log}(\ell)+6)$& $(1-1/k)(\ell\mathrm{log}(\ell)+2\ell+9+k)$ \\
\hline

\end{tabular}
\end{table*}
\begin{table}[hb!]
\centering
\caption{The online running time and amortized communication costs of per party of the sub-protocol. Communication is in MB. Running time is in seconds. $n=21$ is the number of parties. $t$ is the number of adversaries. $k$ is the size of the package in PSS. $\mathrm{Conv}_{32,32,64,5,1}$;$\mathrm{MatMult}_{16,128,8192}$;$\mathrm{ReLU}_{64}$;$\mathrm{Maxpool}_{16,8,2}$.
} 
\label{table::Mirc}
\renewcommand{\arraystretch}{1.2}  
\begin{tabular}{|c|c|c|c|c|c|c|}
\hline
Operation & t &Protocol 
& k
& Comm.(MB)
& WAN(s)
&LAN(s) \\
\hline
\multirow{4}{*}{$\mathrm{Conv}$} & \multirow{2}{*}{3} & LXY24 & - & 0.61 & 0.84 & 0.08\\
& & \textbf{Ours} & 8 & \textbf{0.09} & \textbf{0.15} & \textbf{0.05} \\
\cline{2-7}
 & \multirow{2}{*}{7} & LXY24 & - & 0.80 & 0.87 & 0.30 \\
 & & \textbf{Ours} & 4 & \textbf{0.16} & \textbf{0.27} & \textbf{0.26} \\
\cline{1-7}
\multirow{4}{*}{$\mathrm{MatMult}$} & \multirow{2}{*}{3} & LXY24 & - & 1.30 & 1.84 & 0.19 \\
& & \textbf{Ours} & 8 & \textbf{1.10} & \textbf{0.88} & \textbf{0.39}   \\
\cline{2-7}
& \multirow{2}{*}{7} & LXY24 & - & 1.70 & 1.99 & 0.41 \\
& & \textbf{Ours} & 4 & \textbf{1.16} & \textbf{0.92} & \textbf{0.41} \\
\cline{1-7}
\multirow{4}{*}{$\mathrm{ReLU}$} & \multirow{2}{*}{3} & LXY24 & - & 14.61 & 11.27 & 1.73 \\
& & \textbf{Ours} & 8 & \textbf{5.31} & \textbf{7.42} & \textbf{1.40} \\
\cline{2-7}
 & \multirow{2}{*}{7} & LXY24 & - & 16.96 & 12.62 &1.55 \\
& & \textbf{Ours} & 4 & \textbf{7.75} & \textbf{10.33} & \textbf{2.13} \\
\cline{1-7}
\multirow{4}{*}{$\mathrm{Maxpool}$} & \multirow{2}{*}{3} & LXY24 & - & 19.25 & 14.57 & 2.51 \\
& & \textbf{Ours} & 8 & \textbf{7.00} & \textbf{10.09} & \textbf{2.06} \\
\cline{2-7}
 & \multirow{2}{*}{7} & LXY24 & - & 22.35 & 16.86 & 3.85 \\
& & \textbf{Ours} & 4 & \textbf{10.22} & \textbf{13.66} & \textbf{2.79} \\
\hline
\end{tabular}
\end{table}

\subsection{Microbenchmarks}
\label{Mircben}
We present the theoretical online round and communication complexity of the main protocols in Table~\ref{table::TROur}, and benchmark the online costs of the main sub-protocols for secure inference, including:
\begin{itemize}
    \item Matrix multiplication $\mathrm{MatMult}_{u,v,o}$ with truncation between a $u \times v$ matrix and a $v \times o$ matrix.
    \item 2D convolution $(\mathrm{Conv}_{m,c,o, f,s})$ with the inputs of size $m \times m$, $c$ input channels, $o$ kernels of size $f \times f$ , and the stride $s$.
    \item Secure ReLU $(\mathrm{ReLU}_{n})$ for a matrix of size $n \times n$.
    \item Maxpool $(\mathrm{Maxpool}_{m,c,f})$ with $m \times m $ inputs, $c$ input channels, and $f\times f$ window.
\end{itemize}

Concretely, we measure the running time and communication costs in both WAN and LAN settings with $n=21$ parties, where the size of $k$ is chosen randomly. The results are presented in Table~\ref{table::Mirc}. The offline costs of the main sub-protocols are given in Appendix~\ref{appendixD}.

For the convolution operations, we reduce the communication costs by $5.0-6.8\times$ compared to LXY24. The improvements arise from the use of our filter packing technique. Furthermore, we are $3.2-5.6\times$ faster in WAN and $1.15-1.6\times$ faster in LAN. We observe that the advantage increases as $k$ increases. For the matrix multiplication operations, we reduce the communication costs by $1.2-1.5\times$ compared to LXY24. This stems from our efficient vector-matrix multiplication. Moreover, we are $2.1\times$ faster in WAN. It should be noted that due to computational cost, our method offers little advantage in LAN. For the ReLU and Maxpool operations, we reduce the communication costs by $2.2-2.8\times$ and achieve $1.2-1.5\times$ speedups in the WAN setting,  compared to LXY24. This benefits from our parallelizable non-linear protocol. However, in the LAN setting, the advantage becomes less significant.

\begin{table*}[ht!]
\centering
\caption{The amortized communication costs of per party of Liu et al.LXY24 and our protocol over different datasets and networks for secure inference. Communication is in MB. $n$ is the number of parties. $t$ is the number of adversaries. $k$ is the size of the package in PSS, and OOM indicates memory overflow.
}
\label{table::costcomputaion}
\renewcommand{\arraystretch}{1.2}  
\begin{tabular}{|c|c|c|c|ccc|ccc|ccc|ccc|}
\hline
\multirow{2}{*}{n} & \multirow{2}{*}{t} & \multirow{2}{*}{protocol} & \multirow{2}{*}{k} 
& \multicolumn{3}{c|}{MiniONN(MNIST)} 
& \multicolumn{3}{c|}{LeNet(MNIST)} 
& \multicolumn{3}{c|}{AlexNet(CIFAR10)}
& \multicolumn{3}{c|}{VGG16(CIFAR10)}\\
\cline{5-16}
& & & & offline & online & total
 & offline & online & total & offline & online & total & offline & online & total\\
\hline
\multirow{4}{*}{11} & \multirow{2}{*}{3} & LXY24 & - & 71.77 & 40.45 & 112.22 & 105.63 & 59.53 & 165.16 & 165.04 & 99.40 & 264.44 & 1975.53 & 1113.65 & 3089.18 \\
& & \textbf{Ours} & 3 & \textbf{46.36} & \textbf{21.42} & \textbf{67.78} & \textbf{63.28} & \textbf{29.24} & \textbf{92.52} & \textbf{89.54} & \textbf{46.92} & \textbf{136.46} & \textbf{1157.48} & \textbf{534.87} & \textbf{1692.35} \\
\cline{2-16}
& \multirow{2}{*}{4} & LXY24 & - & 67.81 & 43.17 & 110.98 & 99.81 & 63.54 & 163.35 & 154.42 & 106.07 & 260.49 & 1866.63 &  1188.51 & 3055.14 \\
& & \textbf{Ours} & 2 &\textbf{55.78} & \textbf{26.59} & \textbf{82.37} & \textbf{82.10} & \textbf{39.14} & \textbf{121.24} & \textbf{121.60} & \textbf{65.41} & \textbf{187.01} & \textbf{1537.66} & \textbf{732.92} & \textbf{2270.58} \\
\hline

\multirow{4}{*}{21} & \multirow{2}{*}{3} &LXY24 & - & 78.73 & 36.97 & 115.70 & 115.89 & 54.42 & 170.31 & 183.03 & 90.86 & 273.89 & 2167.38 &1017.94 &3185.32\\
& & \textbf{Ours} & 8 &\textbf{22.63} & \textbf{8.32} & \textbf{30.95} & \textbf{39.20} & \textbf{14.40} & \textbf{53.60} & \textbf{47.93} & \textbf{20.44} & \textbf{68.37} & \textbf{624.24} & \textbf{229.78} & \textbf{854.02}\\
\cline{2-16}
& \multirow{2}{*}{7} &LXY24 & - & 72.29 & 42.93 & 115.22 & 106.40 & 63.20 & 169.60 & 169.30 & 105.50 & 270.80 & 1989.81 & 1182.10 & 3171.91\\
& & \textbf{Ours} & 4 &\textbf{34.02} & \textbf{14.49} & \textbf{48.51} & \textbf{50.51} & \textbf{21.50} & \textbf{72.01} & \textbf{73.15} & \textbf{35.63} & \textbf{108.78} & \textbf{938.21} & \textbf{399.58} & \textbf{1337.79}\\
\hline

\multirow{4}{*}{31} & \multirow{2}{*}{3} & LXY24 & - & 80.89 & 35.62 & 116.51 & 119.07 & 52.42 & 171.49 & 188.64 & 87.54 & 276.18 & 2226.84 & 980.78 & 3207.62 \\
& & \textbf{Ours} & 13 &\textbf{29.13} & \textbf{8.74} & \textbf{37.87} & \textbf{32.80} & \textbf{9.82} & \textbf{42.62} & \textbf{40.32} & \textbf{14.23} & \textbf{54.55} & \textbf{505.03} & \textbf{151.77} & \textbf{656.80} \\
\cline{2-16}
& \multirow{2}{*}{12} &LXY24 & - & 70.30 & 44.85 & 115.15 & 103.47 & 66.02 & 169.49 & 159.74 & 110.21 & 269.95 & OOM & OOM& OOM\\
& & \textbf{Ours} & 4 &\textbf{33.95} & \textbf{14.20} & \textbf{48.15} & \textbf{50.39} & \textbf{21.08} & \textbf{71.47} & \textbf{72.84} & \textbf{34.92} & \textbf{107.76} & \textbf{936.08} &\textbf{391.66} & \textbf{1327.74} \\
\hline

\multirow{4}{*}{63} & \multirow{2}{*}{3} &LXY24 & - & 83.06 & 34.11 & 117.17 & 122.26 & 50.21 & 172.47 & 194.29 & 83.85 & 278.14 & OOM & OOM& OOM\\
& & \textbf{Ours} & 29 &\textbf{24.86} & \textbf{4.68} & \textbf{29.54} & \textbf{28.04} & \textbf{5.25} & \textbf{33.29} & \textbf{33.23} & \textbf{7.51} & \textbf{40.74} & \textbf{458.81} &\textbf{86.83} & \textbf{545.64} \\
\cline{2-16}
& \multirow{2}{*}{20} &LXY24 & - & 75.26 & 42.84 & 118.10 & 110.78 & 63.05 & 173.83 & 172.48 & 105.27 & 277.75 &OOM & OOM& OOM \\
& & \textbf{Ours} & 12 & \textbf{27.44} & \textbf{7.70} & \textbf{35.14} & \textbf{32.27} & \textbf{9.05} & \textbf{41.32} & \textbf{38.49} & \textbf{12.71} & \textbf{51.20} & \textbf{542.71} & \textbf{152.61} & \textbf{695.32} \\
\hline
\end{tabular}
\end{table*}

\begin{table*}
\renewcommand{\arraystretch}{1.2}  
\centering
\caption{Running time of Liu et al.LXY24 and our protocol over different datasets and networks for secure inference in LAN. Running time is in seconds. $n$ is the number of parties. $t$ is the number of adversaries. $k$ is the size of the package in PSS, and OOM indicates memory overflow.
}
\label{table::computaiontimeLAN}

\begin{tabular}{|c|c|c|c|ccc|ccc|ccc|ccc|}
\hline
\multirow{2}{*}{n} & \multirow{2}{*}{t} & \multirow{2}{*}{protocol} & \multirow{2}{*}{k} 
& \multicolumn{3}{c|}{MiniONN(MNIST)} 
& \multicolumn{3}{c|}{LeNet(MNIST)} 
& \multicolumn{3}{c|}{AlexNet(CIFAR10)}
& \multicolumn{3}{c|}{VGG16(CIFAR10)}\\
\cline{5-16}
& & & & offline & online & total
 & offline & online & total & offline & online & total & offline & online & total\\
\hline
\multirow{4}{*}{11} & \multirow{2}{*}{3} & LXY24 & - & 1.13 & 0.20 &	1.33 	& 1.70 &	0.28 	& 1.98 &	2.74 &	0.46 	& 3.20 	& 39.30 	& 5.12 &	44.42 
 \\
& & \textbf{Ours} & 3 & \textbf{0.93} & \textbf{0.33} & \textbf{1.26} & \textbf{1.29} & \textbf{0.46} & \textbf{1.75} & \textbf{1.77} & \textbf{0.74} & \textbf{2.51} & \textbf{29.26} & \textbf{8.26} & \textbf{37.52} \\
\cline{2-16}
& \multirow{2}{*}{4} & LXY24 & - & 1.48 &	0.21 &	1.69 &	2.22 &	0.31 &	2.53 	&3.48 &	0.51 	&3.99 &	48.96 	&5.83 &	54.79 
 \\
& & \textbf{Ours} & 2 &\textbf{1.27} & \textbf{0.47} & \textbf{1.74} & \textbf{1.75} & \textbf{0.68} & \textbf{2.43} & \textbf{2.66} & \textbf{1.17} & \textbf{3.83} & \textbf{42.37} & \textbf{13.06} & \textbf{55.43} \\
\hline

\multirow{4}{*}{21} & \multirow{2}{*}{3} &LXY24 & - & 1.17 &	0.23 &	1.40 	& 1.69 	& 0.29 &	1.98 	& 2.79 &	0.57 &	3.36 &	23.22 &	6.14 &	29.36 
\\
& & \textbf{Ours} & 8 &\textbf{0.74} & \textbf{0.34} & \textbf{1.08} & \textbf{1.37} & \textbf{0.60} & \textbf{1.97} & \textbf{1.54} & \textbf{0.85} & \textbf{2.39} & \textbf{21.33} & \textbf{9.83} & \textbf{31.16}\\
\cline{2-16}
& \multirow{2}{*}{7} & LXY24 & - & 1.93 &	0.28 &	2.21 &	2.84 &	0.44& 	3.28 &	4.78 &	0.69 &	5.47 &	38.62 &	8.68 &	47.30 
\\
& & \textbf{Ours} & 4 &\textbf{1.10} & \textbf{0.45} & \textbf{1.55} & \textbf{1.71} & \textbf{0.72} & \textbf{2.43} & \textbf{2.21} & \textbf{1.19} & \textbf{3.40} & \textbf{35.01} & \textbf{13.25} & \textbf{48.26}\\
\hline

\multirow{4}{*}{31} & \multirow{2}{*}{3} & LXY24 & - & 1.15 	& 0.25 	& 1.40 &	1.76 &	0.36 &	2.12 	& 2.68 &	0.89 &	3.57 & 	38.25 &	7.09 &	45.34 
 \\
& & \textbf{Ours} & 13 &\textbf{1.37} & \textbf{0.68} & \textbf{2.05} & \textbf{1.59} & \textbf{0.79} & \textbf{2.35} & \textbf{1.87} & \textbf{1.15} & \textbf{3.02} & \textbf{25.83} & \textbf{13.63} & \textbf{39.46} \\
\cline{2-16}
& \multirow{2}{*}{12} & LXY24 & - & 3.15 &	0.44 &	3.59 &	4.79& 	0.64 	&5.43 &	7.54 &	1.35 &	8.89 
 & OOM & OOM& OOM\\
& & \textbf{Ours} & 4 &\textbf{1.56} & \textbf{0.68} & \textbf{2.24} & \textbf{2.34} & \textbf{1.24} & \textbf{3.58} & \textbf{3.28} & \textbf{1.62} & \textbf{4.90} & \textbf{51.42} &\textbf{21.59} & \textbf{73.01} \\
\hline

\multirow{4}{*}{63} & \multirow{2}{*}{3} &LXY24 & - & 1.77 &	0.86 	& 2.63 &	2.76 &	0.51 	& 3.27 &	3.91 &	1.97 &	5.88 
 & OOM & OOM& OOM\\
& & \textbf{Ours} & 29 &\textbf{2.65} & \textbf{1.72} & \textbf{4.37} & \textbf{3.06} & \textbf{1.93} & \textbf{4.99} & \textbf{3.52} & \textbf{2.78} & \textbf{6.30} & \textbf{50.31} &\textbf{33.26} & \textbf{83.57} \\
\cline{2-16}
& \multirow{2}{*}{20} &LXY24 & - & 4.68 &	1.16 &	5.84 &	7.27 &	1.13& 	8.40& 	12.59 &	2.61& 	15.20 
&OOM & OOM& OOM \\
& & \textbf{Ours} & 12 & \textbf{2.42} & \textbf{1.50} & \textbf{3.92} & \textbf{2.83} & \textbf{1.85} & \textbf{4.68} & \textbf{3.35} & \textbf{2.48} & \textbf{5.83} & \textbf{82.09} & \textbf{57.29} & \textbf{139.38} \\
\hline
\end{tabular}
\end{table*}

\begin{table*}[t]
\centering
\caption{Running time of LXY24 and our protocol over different datasets and networks for secure inference in WAN. Running time is in seconds. $n$ is the number of parties. $t$ is the number of adversaries. $k$ is the size of the package in PSS, and OOM indicates memory overflow.
}
\label{table::computaiontimeWAN}
\renewcommand{\arraystretch}{1.2}  
\begin{tabular}{|c|c|c|c|ccc|ccc|ccc|ccc|}
\hline
\multirow{2}{*}{n} & \multirow{2}{*}{t} & \multirow{2}{*}{protocol} & \multirow{2}{*}{k} 
& \multicolumn{3}{c|}{MiniONN(MNIST)} 
& \multicolumn{3}{c|}{LeNet(MNIST)} 
& \multicolumn{3}{c|}{AlexNet(CIFAR10)}
& \multicolumn{3}{c|}{VGG16(CIFAR10)}\\
\cline{5-16}
& & & & offline & online & total
 & offline & online & total & offline & online & total & offline & online & total\\
\hline
\multirow{4}{*}{11} & \multirow{2}{*}{3} &LXY24 & - & 17.30 & 9.64& 	26.94 &	23.80 &	13.92 & 37.72 &	35.35 &	22.15 &	57.50 &	426.93 & 182.71 & 609.64 \\
& & \textbf{Ours} & 3 & \textbf{18.63} & \textbf{13.52} & \textbf{32.15} & \textbf{24.04} & \textbf{17.06} & \textbf{41.10} & \textbf{34.93} & \textbf{30.86} & \textbf{65.79} & \textbf{396.67} & \textbf{261.88} & \textbf{658.55} \\
\cline{2-16}
& \multirow{2}{*}{4} &LXY24 & - & 16.82 & 9.93 & 26.75 & 23.51 &	13.21 & 36.72 & 34.34 & 22.43 & 56.77 & 420.59 & 183.26 & 603.85 
 \\
& & \textbf{Ours} & 2 &\textbf{19.11} & \textbf{15.21} & \textbf{34.32} & \textbf{27.50} & \textbf{21.76} & \textbf{49.26} & \textbf{39.24} & \textbf{37.38} & \textbf{76.62} & \textbf{451.03} & \textbf{335.75} & \textbf{786.78} \\
\hline

\multirow{4}{*}{21} & \multirow{2}{*}{3} &LXY24 & - & 32.82 & 15.68 	& 48.50 & 48.14 & 22.02 & 70.16 & 64.87 & 35.80 & 100.67 & 914.20 & 307.68 & 1221.88 
\\
& & \textbf{Ours} & 8 &\textbf{23.08} & \textbf{11.84} & \textbf{34.92} & \textbf{39.41} & \textbf{18.66} & \textbf{58.07} & \textbf{46.28} & \textbf{29.24} & \textbf{75.52} & \textbf{574.19} & \textbf{239.57} & \textbf{813.76}\\
\cline{2-16}
& \multirow{2}{*}{7} &LXY24 & - & 30.16 & 15.72 & 45.88&44.11&22.05 	&66.16 	&62.81 &36.66 &99.47 &910.92 &357.59 &1268.51 \\
& & \textbf{Ours} & 4 &\textbf{27.76} & \textbf{16.93} & \textbf{44.69} & \textbf{39.87} & \textbf{23.74} & \textbf{63.61} & \textbf{53.99} & \textbf{40.58} & \textbf{94.57} & \textbf{684.68} & \textbf{360.72} & \textbf{1045.40}\\
\hline

\multirow{4}{*}{31} & \multirow{2}{*}{3} &LXY24 & - & 48.14&20.09 	&68.23 &67.59 &28.47&96.06&95.06 & 46.27& 141.33&1203.35&461.91&1665.26 \\
& & \textbf{Ours} & 13 &\textbf{49.91} & \textbf{17.84} & \textbf{67.75} & \textbf{55.79} & \textbf{19.59} & \textbf{75.38} & \textbf{65.68} & \textbf{30.40} & \textbf{96.08} & \textbf{793.62} & \textbf{240.29} & \textbf{1033.91} \\
\cline{2-16}
& \multirow{2}{*}{12} & LXY24 & - & 44.12&22.01&66.13&63.73&30.32& 	94.05&90.15&49.95&140.10 
& OOM & OOM& OOM\\
& & \textbf{Ours} & 4 &\textbf{39.76} & \textbf{22.66} & \textbf{62.42} & \textbf{57.29} & \textbf{31.31} & \textbf{88.60} & \textbf{77.38} & \textbf{52.89} & \textbf{130.27} & \textbf{1029.76} &\textbf{498.14} & \textbf{1527.90} \\
\hline

\multirow{4}{*}{63} & \multirow{2}{*}{3} & LXY24 & - & 89.99 & 36.69 &	126.68&131.76&	55.08&	186.84 &	188.09 &	86.27 &	274.36 
& OOM & OOM& OOM\\
& & \textbf{Ours} & 29 &\textbf{102.48} & \textbf{20.09} & \textbf{122.57} & \textbf{112.61} & \textbf{22.09} & \textbf{134.70} & \textbf{129.51} & \textbf{33.07} & \textbf{162.58} & \textbf{1738.20} &\textbf{343.99} & \textbf{2082.19} \\
\cline{2-16}
& \multirow{2}{*}{20} &LXY24 & - & 87.73 & 45.19 & 132.92&124.20 &64.03& 	188.23 &	180.91 &	106.59& 	287.50 
&OOM & OOM& OOM \\
& & \textbf{Ours} & 12 & \textbf{89.40} & \textbf{26.46} & \textbf{115.86} & \textbf{103.31} & \textbf{30.16} & \textbf{133.47} & \textbf{119.13} & \textbf{44.83} & \textbf{163.96} & \textbf{1791.95} & \textbf{476.84} & \textbf{2268.79} \\
\hline
\end{tabular}
\end{table*}

\subsection{Benchmarks on Secure Inference}
 We conducted evaluation on 4 standard neural networks: MiniONN~\cite{Liu2017ObliviousNN} from the secure inference community, and LeNet~\cite{LeCun1998GradientbasedLA}, AlexNet~\cite{Krizhevsky2012ImageNetCW}, and VGG16~\cite{Simonyan2014VeryDC} from the deep learning community. We select 2 standard benchmarking datasets: MNIST~\cite{LeCun1998GradientbasedLA} and CIFAR-10~\cite{Krizhevsky2009LearningML}.

\subsubsection{Communication Costs of Secure Inference}
The communication costs of the protocol from LXY24 and our protocol on various datasets and neural networks are shown in Table~\ref{table::costcomputaion}. Compared to LXY24, we achieve a significant advantage in terms of offline communication, online communication, and total communication costs. Moreover, our analysis reveals a progressively more pronounced advantage along the diagonal from the upper-left to the lower-right corner of the table. For instance, compared to LXY24, our approach improves the communication $1.3-5.9\times$, $1.6-11.2\times$, and $1.4-6.8\times$  in offline, online, and total communication overhead respectively as $k$ increases for AlexNet. Additionally, for VGG16, LXY24 suffers from memory overflow under both the 31PC where 12 parties are controlled by the adversary and 63PC, whereas our approach continues to execute successfully. This advantage primarily stems from the fact that more parties enable the selection of larger values for $k$ in PSS, where an increased $k$ confers greater benefits for deeper neural networks and larger datasets.

\subsubsection{Running Time of Secure Inference (WAN)} The running time of LXY24 and ours in WAN is shown in Table~\ref{table::computaiontimeWAN}. Similar to Table~\ref{table::costcomputaion}, our advantages become increasingly significant from the top-left to the bottom-right of the table. However, it is worth noting that in the case of 11PC, there is no improvement compared to LXY24. This phenomenon primarily stems from our scheme is significantly affected by computational overhead, and when a small $k$ is selected, the communication efficiency improvement is limited. Furthermore, we can find that when the communication overhead increases by more than $3.8\times$, our online running time will be significantly improved. For example, the online phase is $2.4-2.6\times$ faster under 63PC for AlexNet. At the same time, as the number of parties and 
$k$ increase, the offline running time and the total running time can be improved by $1.1-1.6\times$ and $1.1-1.8\times$ respectively. Additionally, ours performs significantly better for deeper networks. For instance, the performance on AlexNet is better than that on MiniONN. The reason is that deeper networks involve more convolution operations, and the improvements in our communication costs and running time become more pronounced. In section~\ref{Mircben}, it can be observed that our method achieves significant advantages in convolution operations.

\subsubsection{Running Time of Secure Inference (LAN)} The running time of LXY24 and ours in LAN is shown in Table~\ref{table::computaiontimeLAN}. We observe that in the online phase, our running time shows little to no advantage compared to LXY24, as our scheme is more heavily affected by computational overhead. However, the offline and the total running time, as the number of parties and $k$ increases, will also have significant improvements, with the increase being $1.1-3.8\times$ and $1.2-2.6\times$ respectively.

In a word, our improvements in the LAN setting are less pronounced compared to those in the WAN setting. The reason lies in the higher computational overhead required by our scheme. Fortunately, as the number of parties and $k$ increases, our scheme exhibits more substantial advantages, mainly owing to our efficient vector-matrix multiplication, parallel convolution computation, and parallel non-linear computation based on PSS. Moreover, MPC protocols are more commonly deployed in WAN settings, so our improvements hold practical significance.

\section{Related Work}
Privacy-preserving deep learning has been a prominent research focus in recent years. Based on the number of participating parties, existing works can be broadly categorized into two categories: those involving a small number of parties (e.g., 2PC, 3PC, 4PC) and those designed for a large number of parties ($n$PC).

In the area of two-party computation (2PC), ABY~\cite{Demmler2015ABYA} and ABY2.0~\cite{Patra2020ABY20IM} design a framework that combines secure computation schemes based on arithmetic sharing , boolean sharing, and Yao’s garbled circuits (GC). SecureML~\cite{Mohassel2017SecureMLAS} proposes a scheme in the two-server outsourced computation models with secret sharing and GC. Meanwhile, MiniONN~\cite{Liu2017ObliviousNN} designs fast matrix multiplication protocols. Gazelle~\cite{Juvekar2018GazelleAL} and Delphi~\cite{Mishra2020DelphiAC} combine techniques from HE and MPC to achieve fast private inference. CrypTFlow2~\cite{Rathee2020CrypTFlow2P2} designs an efficient comparison protocol, and its followup works~\cite{Huang2022CheetahLA,Feng2025PantherPS, Balla2025SecONNdsSO} focus on improving performance. In addition, there exists many works~\cite{Pang2024BOLTPA,Lu2023BumbleBeeST,Zhang2024SecureTI} for large language model inference combines techniques from HE and secret sharing. In order to solve the performance bottleneck of 2PC, Chameleon~\cite{Riazi2018ChameleonAH} employed a semi-honest third-party to generate correlated randomness for multiplication triplets in offline. Further, for improving performance in non-linear functions such as ReLU, SecureNN~\cite{Wagh2019SecureNN3S} constructs novel protocols for comparsion with the help of a third-party. For providing better overall efficiency, some works ABY3~\cite{Mohassel2018ABY3AM}, Falcon\cite{Wagh2020FalconHM}, Swift~\cite{Koti2020SWIFTSA} are propoed based on 3PC replicated secret sharing. Additionally, ~\cite{Chaudhari2019ASTRAHT,Patra2020BLAZEBF} and~\cite{Dalskov2020FantasticFH,Byali2020FLASHFA,Rachuri2019TridentE4,Koti2021TetradAS} focus on resisting malicious adversaries and require much more costs in 3PC and 4PC, respectively.

Instead, a critical limitation common to all these representative works lies in their restriction to a maximum of four participating parties, consequently exhibiting inadequate scalability. In order to solve the bottleneck, many works are proposed. Motion~\cite{Braun2022MOTIONA} designs an $n$PC framework that
combines secure computation schemes based on arithmetic
sharing, boolean sharing, and GC. Howerver, since computation is required between every pair of parties in Motion, the overall efficiency is limited.  Alessandro et al.~\cite{Baccarini2023MultiPartyRS} propose an $n$PC framework by utilizing replicated
secret sharing, but its sharing structure is relatively complex and incurs high communication overhead for multiplication operations. LXY24 first constructs a practical inference framework using Shamir secret sharing. Unfortunately, when dealing with deep neural networks such as VGG16, the communication overhead remains substantial. Franklin et al.~\cite{franklin1992communication} introduced the PSS and its follow-up works~\cite{Goyal2021UnconditionalCM, goyal2022sharing,Escudero2022TurboPackHM} focus on improving multiplication operation and resisting malicious adversaries in theory. In this work, we mainly achieve efficient vector-matrix multiplication, convolution computation and non-linear operations for more practical PPDL using PSS.

\section{Conclusion \& Future Work}
In this paper, we construct secure blocks based on packed Shamir secret sharing for high-throughput and scalable privacy-preserving deep learning. Extensive evaluations also present our improvements. For future work, we are willing to extend to malicious adversary scenarios and improve computational efficiency by GPU.

\appendix

\subsection{The Offline Costs of Microbenchmarks}
\label{appendixD}
The offline costs of the main sub-protocols are given in Table~\ref{table::MircOff}. Concretely, for linear operations (i.e., convolution and matrix multiplication operations), there is no advantage in running time. The reason is that the offline complexity for generating random bits and truncation triples is greater. For example, generating truncation triples needs to call the degree transformation protocol (Algorithm~\ref{DegreeTrans}) multiple times. However, it is worth noting that our communication costs still show improvement with a larger $k$. Take the case with $k=8$ as an example, the offline cost of matrix multiplication achieves improvement. Therefore, we can reasonably infer that when a larger $k$ is chosen, the communication cost is further reduced, which in turn leads to an advantage in running time. This is reflected in the Table~\ref{table::computaiontimeWAN}. For the non-linear operations (i.e., ReLU and Maxpool), we reduce the communication costs by $3.1-3.5\times$ compared to LXY24. Furthermore, we are $1.5-2.0\times$ faster in WAN and $1.1-1.6\times$ faster in LAN. 

\begin{table}[ht!]
\centering
\caption{The offline running time and amortized communication costs of per party of the sub-protocol. Communication is in MB. Running time is in seconds. $n=21$ is the number of parties. $t$ is the number of adversaries. $k$ is the size of the package in PSS. $\mathrm{Conv}_{32,32,64,5,1}$;$\mathrm{MatMult}_{16,128,8192}$;$\mathrm{ReLU}_{64}$;$\mathrm{Maxpool}_{16,8,2}$.
} 
\label{table::MircOff}
\renewcommand{\arraystretch}{1.2}  
\begin{tabular}{|c|c|c|c|c|c|c|}
\hline
Operation & t &Protocol 
& k
& Comm.(MB)
& WAN(s)
&LAN(s) \\
\hline
\multirow{4}{*}{$\mathrm{Conv}$} & \multirow{2}{*}{3} & LXY24 & - & 65.74 & 65.66 & 7.54\\
& & \textbf{Ours} & 8 & \textbf{62.12} & \textbf{72.28} & \textbf{10.67} \\
\cline{2-7}
 & \multirow{2}{*}{7} & LXY24 & - & 77.18 & 72.41 & 10.04 \\
 & & \textbf{Ours} & 4 & \textbf{76.29} & \textbf{83.19} & \textbf{11.92} \\
\cline{1-7}
\multirow{4}{*}{$\mathrm{MatMult}$} & \multirow{2}{*}{3} & LXY24 & - & 141.94 & 140.27 & 16.32 \\
& & \textbf{Ours} & 8 & \textbf{117.49} & \textbf{203.83} & \textbf{32.37}   \\
\cline{2-7}
& \multirow{2}{*}{7} & LXY24 & - & 165.97 & 154.22 & 21.49 \\
& & \textbf{Ours} & 4 & \textbf{205.77} & \textbf{225.00} & \textbf{33.02} \\
\cline{1-7}
\multirow{4}{*}{$\mathrm{ReLU}$} & \multirow{2}{*}{3} & LXY24 & - & 26.81 & 15.73 & 2.47 \\
& & \textbf{Ours} & 8 & \textbf{7.68} & \textbf{8.49} & \textbf{1.56} \\
\cline{2-7}
 & \multirow{2}{*}{7} & LXY24 & - & 23.49 &16.50 & 2.44\\
& & \textbf{Ours} & 4 & \textbf{11.04} & \textbf{11.24} & \textbf{2.17} \\
\cline{1-7}
\multirow{4}{*}{$\mathrm{Maxpool}$} & \multirow{2}{*}{3} & LXY24 & - & 35.34 & 22.31 & 2.92  \\
& & \textbf{Ours} & 8 & \textbf{10.12} & \textbf{11.12} & \textbf{2.01} \\
\cline{2-7}
 & \multirow{2}{*}{7} & LXY24 & - & 30.97 & 21.48 & 3.74 \\
& & \textbf{Ours} & 4 & \textbf{14.55} & \textbf{14.82} & \textbf{2.53} \\
\hline
\end{tabular}
\end{table}

\subsection{The protocol \texorpdfstring{$\Pi_{Random}$}{Random},\texorpdfstring{$\Pi_{RandomBits}$}{RandomBits},\texorpdfstring{$\Pi_{DegreeTrans}$}{DegreeTrans},\texorpdfstring{$\Pi_{Xor}$}{Xor}}
\label{appendixB}
\subsubsection{\texorpdfstring{$\Pi_{Random}$}{Random}} Batch generation of random shares using Vandermonde matrices. The protocol implementation in Algorithm~\ref{Random}. It is a natural
 extension of protocols for generating random shares based on Shamir secret sharing. Note that the degree of the polynomial chosen in Algorithm~\ref{Random} is $d$. We can select any degree $d_1$ as long as it satisfies $d\leq d_1 \leq n-1$. The communication complexity is 1 round with $\frac{n-1}{k(n-t)}$ field elements amortized per generated share. 

\begin{algorithm}
    \caption{Random Shares $\Pi_{Random}$}
    \KwIn{None.}
    \KwOut{ $\llbracket \boldsymbol{r}_1\rrbracket_d,\cdots,\llbracket \boldsymbol{r}_{n-t}\rrbracket_d$.}

    \begin{enumerate}[label=\arabic*., align=left, leftmargin=1.2em, labelwidth=1em, labelsep=0.2em]
        \item Each party $P_i$ chooses a random vector $\boldsymbol{u}^i\in\mathbb{F}_{p}^k$ and distributs $\llbracket\boldsymbol{u}^i\rrbracket_d$ to other parties.
        \item Holding $n$ shares $(\llbracket \boldsymbol{u}^1\rrbracket_d,\cdots,\llbracket \boldsymbol{u}^n\rrbracket_d),$ each party $P_i$ computes $(\llbracket \boldsymbol{r}_1\rrbracket_d,\cdots,\llbracket \boldsymbol{r}_{n-t}\rrbracket_d)=Van(n,n-t)^T\cdot(\llbracket \boldsymbol{u}^1\rrbracket_d,\cdots,\llbracket \boldsymbol{u}^n\rrbracket_d)$.
    \end{enumerate}
    \Return $\llbracket \boldsymbol{r}_1\rrbracket_d,\cdots,\llbracket \boldsymbol{r}_{n-t}\rrbracket_d$.
    \label{Random}
\end{algorithm}

\subsubsection{\texorpdfstring{$\Pi_{DegreeTrans}$}{DegreeTrans}}~\cite{goyal2022sharing} designed a sharing transformation protocol to securely perform the degree transformation. Specifically, given a PSS $\llbracket \boldsymbol{x}\rrbracket_{n-1}$, it outputs $\llbracket \boldsymbol{x}\rrbracket_{n-k}.$ We extend this protocol simply to support arbitrary degree conversions for Shamir secret sharing and PSS in Algorithm~\ref{DegreeTrans}. Note that Step 1 can be performed during the offline phase. The communication complexity of the protocol is 2 round with  $\frac{2(n-1)}{k}(\frac{1}{n-t}+\frac{1}{n})$ field elements amortized per party (resp. $2(n-1)(\frac{1}{n-t}+\frac{1}{n})$ field elements for Shamir secret sharing).

\begin{algorithm}
    \caption{Degree Trans $\Pi_{DegreeTrans}$}
    \KwIn{$\llbracket \boldsymbol{x}\rrbracket_{d_1}$ (reps. $[x|_s]_{d_2}$), $d\leq d_1\leq n-1,d-k+1\leq d_2\leq n-1$.}
    \KwOut{ $\llbracket \boldsymbol{x}\rrbracket_{d_3}$ (reps. $[x|_s]_{d_4}$), $d\leq d_3\leq n-1$, $d-k+1\leq d_4\leq n-1,d_1\neq d_3,d_2\neq d_4$.}

    \begin{enumerate}[label=\arabic*., align=left, leftmargin=1.2em, labelwidth=1em, labelsep=0.2em]
        \item  Get $(\llbracket \boldsymbol{r}\rrbracket_{d_1},\llbracket \boldsymbol{r}\rrbracket_{d_3})$ (resp. $([r|_s]_{d_2},[r|_s]_{d_4})$)$\leftarrow\mathcal{F}_{Random}$.
        \item All parties compute $\llbracket\boldsymbol{x'}\rrbracket_{d_1}=\llbracket\boldsymbol{x}\rrbracket_{d_1}+\llbracket\boldsymbol{r}\rrbracket_{d_1}$ (resp. $[x'|_s]_{d_2}=[x|_s]_{d_2}+[r|_s]_{d_2}$) and send their shares to $P_1$.
        \item $P_1$ reconstructs the secret $\boldsymbol{x}'$ (resp. $x'$) and distributes $\llbracket\boldsymbol{x}'\rrbracket_{d_3}$ (resp. $[x'|_s]_{d_4}$) to all parties.
        \item All parties locally compute $\llbracket\boldsymbol{x}\rrbracket_{d_3}=\llbracket\boldsymbol{x}'\rrbracket_{d_3}-\llbracket\boldsymbol{r}\rrbracket_{d_3}$ (resp. $[x|_s]_{d_4}=[x'|_s]_{d_4}+[r|_s]_{d_4}$).
    \end{enumerate}
    
    \Return $\llbracket \boldsymbol{x}\rrbracket_{d_3}$ (reps. $[x|_s]_{d_4}$).
    \label{DegreeTrans}
\end{algorithm}

\subsubsection{\texorpdfstring{$\Pi_{RandomBits}$}{RandomBits}}~\cite{damgaard2006unconditionally} proposed a protocol to securely generate the shares of random bits. The parties generate a bits $r\in\{0,1\}$ by computing $S(a)=\frac{\frac{a}{\sqrt{a^2}}+1}{2}$, where $a\in\mathbb{F}_p^*$ is a uniformly random non-zero element. $\sqrt{a^2}$ is formulated to equal the unique element in $[1,(p-1)/2]$ and $a/\sqrt{a^2}$ is $1$ or $-1$. We extend this protocol using PSS in Algorithm~\ref{RandomBits}. The communication complexity is 4 rounds with $\frac{n-1}{k}(\frac{5}{n-t}+\frac{4}{n}+k+1)$ field elements amortized per generated random bits.

\begin{algorithm}
    \caption{Random Bits $\Pi_{RandomBits}$}
    \KwIn{None.}
    \KwOut{ $\llbracket \boldsymbol{r}\rrbracket_{d},\boldsymbol{r}=\{0,1\}^k$.}

    \begin{enumerate}[label=\arabic*., align=left, leftmargin=1.2em, labelwidth=1em, labelsep=0.2em]
        \item Get $\llbracket \boldsymbol{a}\rrbracket_d$ by functionality $\mathcal{F}_{Random}$.
        \item Compute $\llbracket \boldsymbol{a}'\rrbracket_d\leftarrow\mathcal{F}_{PDN}(\llbracket \boldsymbol{a}\rrbracket_d,\llbracket \boldsymbol{a}\rrbracket_d)$.
        \item All parties open $\boldsymbol{a}'$. If there exists a zero in $\boldsymbol{a}'$, then abort. Otherwise, proceed as below.
        \item Compute $\boldsymbol{b}_i=\sqrt{\boldsymbol{a}'_i},i\in[k],\llbracket \boldsymbol{b}\rrbracket_{k-1}$ and $\llbracket \boldsymbol{c}\rrbracket_{d+k-1}=\llbracket\boldsymbol{b}\rrbracket_{k-1}\cdot\llbracket\boldsymbol{a}\rrbracket_{d}$ locally.
        \item Get $\llbracket \boldsymbol{c}\rrbracket_{d}\leftarrow\mathcal{F}_{DegreeTrans}(\llbracket \boldsymbol{c}\rrbracket_{d+k-1})$.
        \item Compute $\llbracket\boldsymbol{r}\rrbracket_d=\frac{\llbracket\boldsymbol{c}\rrbracket_d+1}{2}.$

    \end{enumerate}
    
    \Return $\llbracket \boldsymbol{r}\rrbracket_{d}$.
    \label{RandomBits}
\end{algorithm}

\subsubsection{\texorpdfstring{$\Pi_{Xor}$}{Xor}} Given two inputs $\llbracket \boldsymbol{a}\rrbracket_{d_1},\llbracket \boldsymbol{b}\rrbracket_{d},k-1\leq d_1\leq d$, we get $\llbracket \boldsymbol{a}\rrbracket_{d_1}\oplus\llbracket \boldsymbol{b}\rrbracket_{d}$ by computing $\llbracket \boldsymbol{a}\rrbracket_{d_1}\oplus\llbracket \boldsymbol{b}\rrbracket_{d}=\llbracket \boldsymbol{a}\rrbracket_{d_1}+\llbracket \boldsymbol{b}\rrbracket_{d}-2\llbracket \boldsymbol{a}\rrbracket_{d_1}\llbracket \boldsymbol{b}\rrbracket_{d}$, where $\boldsymbol{a},\boldsymbol{b}$ are bit strings. The protocol implementation in Algorithm~\ref{Xor}. The communication complexity is the same as protocol $\Pi_{DegreeTrans}$ amortized per xor operation.

\begin{algorithm}
    \caption{Xor $\Pi_{Xor}$}
    \KwIn{$\llbracket \boldsymbol{a}\rrbracket_{d_1},\llbracket \boldsymbol{b}\rrbracket_{d},k-1\leq d_1\leq d,\boldsymbol{a},\boldsymbol{b}\in\{0,1\}^k$.}
    \KwOut{ $\llbracket \boldsymbol{a}\oplus\boldsymbol{b}\rrbracket_{d},\boldsymbol{a}\oplus\boldsymbol{b}=(\boldsymbol{a}_0\oplus\boldsymbol{b}_0,\cdots,\boldsymbol{a}_{k-1}\oplus\boldsymbol{b}_{k-1})$.}

    \begin{enumerate}[label=\arabic*., align=left, leftmargin=1.2em, labelwidth=1em, labelsep=0.2em]
        \item Compute $\llbracket\boldsymbol{c}\rrbracket_{d+d_1}=2\cdot\llbracket\boldsymbol{a}\rrbracket_{d_1}\cdot\llbracket\boldsymbol{b}\rrbracket_{d}.$
        \item Set $\llbracket \boldsymbol{a}\oplus\boldsymbol{b}\rrbracket_{d+d_1}=\llbracket\boldsymbol{a}\rrbracket_{d_1}+\llbracket\boldsymbol{b}\rrbracket_{d}+\llbracket\boldsymbol{c}\rrbracket_{d+d_1}.$
        \item Get $\llbracket \boldsymbol{a}\oplus\boldsymbol{b}\rrbracket_d\leftarrow\mathcal{F}_{DegreeTrans}(\llbracket \boldsymbol{a}\oplus\boldsymbol{b}\rrbracket_{d+d_1})$.
    
    \end{enumerate}
    
    \Return $\llbracket \boldsymbol{a}\oplus\boldsymbol{b}\rrbracket_d$.
    \label{Xor}
\end{algorithm}

\subsection{Functionality Description}
\label{appendixA}
\begin{tcolorbox}[colback=white, colframe=black, boxsep=1pt, left=10pt, right=10pt]
{\centering  \textbf{Functionality} $\mathcal{F}_{Random}$\par} 
 
    \textbf{Input:} The functionality receives no inputs. 
    
    \textbf{Output:} Compute the following
   
    \begin{enumerate}[label=\arabic*., align=left, leftmargin=1.2em, labelwidth=1em, labelsep=0.2em]
            \item  Randomly sample a vector $\boldsymbol{r}\in\mathbb{F}_p^k$.
            \item Distribute shares of $\llbracket\boldsymbol{r} \rrbracket_d$ to the parties.
    \end{enumerate}
    
\end{tcolorbox}

\begin{tcolorbox}[colback=white, colframe=black, boxsep=1pt, left=10pt, right=10pt]
{\centering  \textbf{Functionality} $\mathcal{F}_{RandomBits}$\par} 
 
          \textbf{Input:} The functionality receives no inputs.
          
           \textbf{Output:} Compute the following
   
\begin{enumerate}[label=\arabic*., align=left, leftmargin=1.2em, labelwidth=1em, labelsep=0.2em]
            \item Randomly sample a bit string $\boldsymbol{r}\in\mathbb{F}_2^k$.
            \item Distribute shares of $\llbracket\boldsymbol{r} \rrbracket_d$ to the parties.
\end{enumerate}
\end{tcolorbox}

\begin{tcolorbox}[colback=white, colframe=black, boxsep=1pt, left=10pt, right=10pt]
    {\centering  \textbf{Functionality} $\mathcal{F}_{DegreeTrans}$\par} 

 \textbf{Input:} The functionality receives inputs $\llbracket \boldsymbol{a}\rrbracket_{d_1}$ (resp. $[a|_s]_{d_2}$), where $d\leq d_1\leq n -1,d-k+1\leq d_2\leq n -1$.
      \textbf{Output:} Compute the following
   
\begin{enumerate}[label=\arabic*., align=left, leftmargin=1.2em, labelwidth=1em, labelsep=0.2em]
  \item Reconstruct $\boldsymbol{a}\in \mathbb{F}_p^k$ (resp. $a\in \mathbb{F}_p$).
  \item Distribute shares of $\llbracket \boldsymbol{a}\rrbracket_{d_3}$ (resp. $[a|_s]_{d_4}$) to the parties, where $d\leq d_3\leq n -1,d-k+1\leq d_4\leq n -1,d_3\neq d_1,d_4\neq d_2.$
\end{enumerate}
\end{tcolorbox}

\begin{tcolorbox}[colback=white, colframe=black, boxsep=1pt, left=10pt, right=10pt]
{\centering  \textbf{Functionality} $\mathcal{F}_{PDN}$\par} 
     \textbf{Input:} The functionality receives inputs $\llbracket \boldsymbol{a}\rrbracket_{d}$ and $\llbracket \boldsymbol{b}\rrbracket_{d}$.
       
     \textbf{Output:} Compute the following
\begin{enumerate}[label=\arabic*., align=left, leftmargin=1.2em, labelwidth=1em, labelsep=0.2em]
  \item Reconstruct $\boldsymbol{a}\in \mathbb{F}_p^k$ and $\boldsymbol{b}\in \mathbb{F}_p^k$ to compute $\boldsymbol{c}_i=\boldsymbol{a}_i\boldsymbol{b}_i,i\in[k]$.
  \item Distribute shares of $\llbracket \boldsymbol{c}\rrbracket_{d}$ to the parties.
\end{enumerate}
\end{tcolorbox}

\begin{tcolorbox}[colback=white, colframe=black, boxsep=1pt, left=10pt, right=10pt]
{\centering  \textbf{Functionality} $\mathcal{F}_{PMatMult}^{Fixed}$\par} 
     \textbf{Input:} The functionality receives inputs $\llbracket A\rrbracket_{d}\in\mathbb{F}_{p}^{u\times v}$ and $\llbracket B\rrbracket_{d}\in\mathbb{F}_{p}^{v\times w}$.
       
     \textbf{Output:} Compute the following
\begin{enumerate}[label=\arabic*., align=left, leftmargin=1.2em, labelwidth=1em, labelsep=0.2em]
  \item Reconstruct $A^i\in \mathbb{F}_p^{u\times v}$ and $B^i\in \mathbb{F}_p^{v\times w},i\in[k]$ to compute $C'^i=A^i\cdot B^i,i\in[k]$.
  \item Compute $C^i=C'^i_{\alpha,\beta}/2^{\ell_x},i\in[k],\alpha\in[u],\beta\in[w]$.
  \item Distribute shares of $\llbracket C\rrbracket_{d}\in\mathbb{F}_p^{u\times w}$ to the parties.
\end{enumerate}
\end{tcolorbox}

\begin{tcolorbox}[colback=white, colframe=black, boxsep=1pt, left=10pt, right=10pt]
{\centering  \textbf{Functionality} $\mathcal{F}_{VecMatMult}$\par} 
       \textbf{Input:} The functionality receives inputs a vector $\llbracket \boldsymbol{a}\rrbracket_{d}\in\mathbb{F}_q^u$ and a matrix $\llbracket \boldsymbol{A}\rrbracket_{d}\in\mathbb{F}_q^{u\times vk}$.
       
       \textbf{Output:} Compute the following
\begin{enumerate}[label=\arabic*., align=left, leftmargin=1.2em, labelwidth=1em, labelsep=0.2em]

  \item Reconstruct $\boldsymbol{a}\in \mathbb{F}_p^{ku}$ and $\boldsymbol{A}\in \mathbb{F}_p^{ku\times vk}$ to compute $\boldsymbol{c}=\boldsymbol{a}\cdot A$.
  \item Distribute shares of $\llbracket \boldsymbol{c}\rrbracket_{d}\in\mathbb{F}_{q}^{v}$ to the parties.
\end{enumerate}
\end{tcolorbox}

\begin{tcolorbox}[colback=white, colframe=black, boxsep=1pt, left=10pt, right=10pt]
{\centering  \textbf{Functionality} $\mathcal{F}_{VecMatMult}^{Fixed}$\par} 
       \textbf{Input:} The functionality receives inputs a vector $\llbracket \boldsymbol{a}\rrbracket_{d}\in\mathbb{F}_q^u$ and a matrix $\llbracket \boldsymbol{A}\rrbracket_{d}\in\mathbb{F}_q^{u\times vk}$.
       
       \textbf{Output:} Compute the following
\begin{enumerate}[label=\arabic*., align=left, leftmargin=1.2em, labelwidth=1em, labelsep=0.2em]

  \item Reconstruct $\boldsymbol{a}\in \mathbb{F}_p^{ku}$ and $\boldsymbol{A}\in \mathbb{F}_p^{ku\times vk}.$
  \item  Compute $\boldsymbol{c}=\boldsymbol{a}\cdot A$ and compute $\boldsymbol{c}_i'=\boldsymbol{c}_i/2^{\ell_x}$.
  \item Distribute shares of $\llbracket \boldsymbol{c}'\rrbracket_{d}\in\mathbb{F}_{q}^{v}$ to the parties.
\end{enumerate}
\end{tcolorbox}

\begin{tcolorbox}[colback=white, colframe=black, boxsep=1pt, left=10pt, right=10pt]
{\centering  \textbf{Functionality} $\mathcal{F}_{VM-RandTuple}$\par} 
\textbf{Input:} The functionality receives no inputs.

\textbf{Output:} Compute the following
\begin{enumerate}[label=\arabic*., align=left, leftmargin=1.2em, labelwidth=1em, labelsep=0.2em]

  \item Randomly sample a vecotr $\boldsymbol{r}\in\mathbb{F}_p^{kv}$ and compute a vector $\boldsymbol{r}'=(\sum_{i=0}^{k-1}\boldsymbol{r}_i,\cdots,\sum_{i=(v-1)k}^{kv-1}\boldsymbol{r}_i)\in\mathbb{F}_p^v$.
  \item Distribute shares of the random pairs $(\llbracket \boldsymbol{r}\rrbracket_{2d},\llbracket \boldsymbol{r'}\rrbracket_{d})$ to the parties.
\end{enumerate}
\end{tcolorbox}

\begin{tcolorbox}[colback=white, colframe=black, boxsep=1pt, left=10pt, right=10pt]
{\centering  \textbf{Functionality} $\mathcal{F}_{TruncTriple}$\par} 
\textbf{Input:} The functionality receives no inputs.

\textbf{Output:} Compute the following
\begin{enumerate}[label=\arabic*., align=left, leftmargin=1.2em, labelwidth=1em, labelsep=0.2em]

  \item Randomly sample a vecotr $\boldsymbol{r}\in\mathbb{F}_p^{kv}$ and compute a vector $\boldsymbol{r}'=(\frac{\sum_{i=0}^{k-1}\boldsymbol{r}_i}{2^{\ell_x}},\cdots,\frac{\sum_{i=(v-1)k}^{kv-1}\boldsymbol{r}_i}{2^{\ell_x}})\in\mathbb{F}_p^v$.
  \item Distribute shares of the random pairs $(\llbracket \boldsymbol{r}\rrbracket_{2d},\llbracket \boldsymbol{r'}\rrbracket_{d})$ to the parties.
\end{enumerate}
\end{tcolorbox}

\begin{tcolorbox}[colback=white, colframe=black, boxsep=1pt, left=10pt, right=10pt]
{\centering  \textbf{Functionality} $\mathcal{F}_{PackTrans}$\par} 
\textbf{Input:} The functionality receives inputs $\llbracket \boldsymbol{x}\rrbracket_d.$

\textbf{Output:} Compute the following
\begin{enumerate}[label=\arabic*., align=left, leftmargin=1.2em, labelwidth=1em, labelsep=0.2em]

  \item Reconstruct $\boldsymbol{x}=(\boldsymbol{x}_0,\cdots,\boldsymbol{x}_{k-1})$ and set $\boldsymbol{x}^i=(\boldsymbol{x_i},\cdots,\boldsymbol{x_i})\in \mathbb{F}_p^{k},i\in[k]$.
  \item Distribute shares of $\llbracket \boldsymbol{x}^i\rrbracket_{d}$ to the parties.
\end{enumerate}
\end{tcolorbox}

\begin{tcolorbox}[colback=white, colframe=black, boxsep=1pt, left=10pt, right=10pt]
{\centering  \textbf{Functionality} $\mathcal{F}_{PreMult}$\par} 
\textbf{Input:} The functionality receives inputs $\llbracket \boldsymbol{a}^0\rrbracket_d,\cdots,$

$\llbracket \boldsymbol{a}^{\ell-1}\rrbracket_d,$ where $\boldsymbol{a}^i$ is a vector of length $k$, $\boldsymbol{a}^i=(\boldsymbol{a}_0^i,\cdots,\boldsymbol{a}_{k-1}^i).$

\textbf{Output:} Compute the following
\begin{enumerate}[label=\arabic*., align=left, leftmargin=1.2em, labelwidth=1em, labelsep=0.2em]

  \item Reconstruct 
  
  $(\boldsymbol{a}^{0}_0,\cdots,\boldsymbol{a}^{0}_{k-1},\boldsymbol{a}^{1}_0,\cdots,\boldsymbol{a}^{1}_{k-1},\cdots,\boldsymbol{a}^{\ell-1}_0,\cdots,\boldsymbol{a}^{\ell-1}_{k-1})$.
  \item Compute all prefix products $\boldsymbol{a}'^0=(\boldsymbol{a}_0^0,\cdots,\boldsymbol{a}_{k-1}^0),$

$ \boldsymbol{a}'^1=(\boldsymbol{a}_0^0\cdot\boldsymbol{a}_0^1,\cdots,\boldsymbol{a}_{k-1}^0\cdot \boldsymbol{a}_{k-1}^1)$,

$ \boldsymbol{a}'^{\ell-1}=(\Pi_{i=0}^{\ell-1}\boldsymbol{a}_0^{i},\cdots,\Pi_{i=0}^{\ell-1}\boldsymbol{a}_{k-1}^{i})$.
  \item Distribute the shares of these these products $\llbracket\boldsymbol{a}'^i\rrbracket_d,i\in[\ell]$ to the parties.
\end{enumerate}
\end{tcolorbox}

\begin{tcolorbox}[colback=white, colframe=black, boxsep=1pt, left=10pt, right=10pt]
{\centering  \textbf{Functionality} $\mathcal{F}_{PreOR}$\par} 
\textbf{Input:} The functionality receives inputs $\llbracket \boldsymbol{a}^0\rrbracket_d,\cdots,$

$\llbracket \boldsymbol{a}^{\ell-1}\rrbracket_d,$ where $\boldsymbol{a}^i$ is a bit vector of length $k$, $\boldsymbol{a}^i=(\boldsymbol{a}_0^i,\cdots,\boldsymbol{a}_{k-1}^i),\boldsymbol{a}^i_{j}\in\{0,1\}.$

\textbf{Output:} Compute the following
\begin{enumerate}[label=\arabic*., align=left, leftmargin=1.2em, labelwidth=1em, labelsep=0.2em]

  \item Reconstruct the bitwise decomposition of $k$ values $\boldsymbol{a}_j,j\in[k]$ such that $\boldsymbol{a}_j=\sum_{i=0}^{\ell-1}2^i\boldsymbol{a}_j^{i}$.
  \item Compute all prefix-ORs $\boldsymbol{a}'^0=(\boldsymbol{a}_0^0,\cdots,\boldsymbol{a}_{k-1}^0),$

$ \boldsymbol{a}'^1=(\boldsymbol{a}_0^0\vee\boldsymbol{a}_0^1,\cdots,\boldsymbol{a}_{k-1}^0\vee \boldsymbol{a}_{k-1}^1)$,

$ \boldsymbol{a}'^{\ell-1}=(\vee_{i=0}^{\ell-1}\boldsymbol{a}_0^{i},\cdots,\vee_{i=0}^{\ell-1}\boldsymbol{a}_{k-1}^{i})$.
  \item Distribute the shares of these prefix-ORs $\llbracket\boldsymbol{a}'^i\rrbracket_d,i\in[\ell]$ to the parties.
\end{enumerate}
\end{tcolorbox}

\begin{tcolorbox}[colback=white, colframe=black, boxsep=1pt, left=10pt, right=10pt]
{\centering  \textbf{Functionality} $\mathcal{F}_{Xor}$\par} 
\textbf{Input:} The functionality receives inputs $\llbracket \boldsymbol{a}\rrbracket_{d_1},\llbracket \boldsymbol{b}\rrbracket_d,k-1\leq d_1\leq d,$ where $\boldsymbol{a}$ (resp. $\boldsymbol{b}$) is a bit vector of length $k$ and $\boldsymbol{a}_i$ (resp. $\boldsymbol{b}_i$)$\in\{0,1\}.$

\textbf{Output:} Compute the following
\begin{enumerate}[label=\arabic*., align=left, leftmargin=1.2em, labelwidth=1em, labelsep=0.2em]

  \item Reconstruct $\boldsymbol{a},\boldsymbol{b}$ to compute $\boldsymbol{c}_i=\boldsymbol{a}_i\oplus \boldsymbol{b}_i,i\in[k]$.
  \item Distribute shares of $\llbracket \boldsymbol{c}\rrbracket_{d}$ to the parties.
\end{enumerate}
\end{tcolorbox}

\begin{tcolorbox}[colback=white, colframe=black, boxsep=1pt, left=10pt, right=10pt]
{\centering  \textbf{Functionality} $\mathcal{F}_{Bitwise-LT}$\par} 
\textbf{Input:} The functionality receives inputs $k$  values $\boldsymbol{a}_j,j\in[k]$ and $\ell$ shares $\llbracket \boldsymbol{b}^i\rrbracket_d,i\in[\ell]$ where $\boldsymbol{b}^i=\{0,1\}^k,\boldsymbol{b}_j=\sum_{i=0}^{\ell}2^i\boldsymbol{b}_j^i.$

\textbf{Output:} Compute the following
\begin{enumerate}[label=\arabic*., align=left, leftmargin=1.2em, labelwidth=1em, labelsep=0.2em]
    \item Reconstruct $\boldsymbol{b}^i$ to compute $\boldsymbol{b}_j=\sum_{i=0}^{\ell-1}2^i\boldsymbol{b}_j^i,j\in[k].$
    \item Compute $\boldsymbol{c}_j=(\boldsymbol{a}_j<\boldsymbol{b}_j),$ where $c\in\{0,1\}.$
    \item Distribute the shares of $\llbracket \boldsymbol{c}\rrbracket_d$ to the parties.

\end{enumerate}
\end{tcolorbox}

\begin{tcolorbox}[colback=white, colframe=black, boxsep=1pt, left=10pt, right=10pt]
{\centering  \textbf{Functionality} $\mathcal{F}_{DReLU}$\par} 
\textbf{Input:} The functionality receives inputs $\llbracket \boldsymbol{a}\rrbracket_d$.

\textbf{Output:} Compute the following
\begin{enumerate}[label=\arabic*., align=left, leftmargin=1.2em, labelwidth=1em, labelsep=0.2em]

  \item Reconstruct $\boldsymbol{a}=(\boldsymbol{a}_0,\cdots,\boldsymbol{a}_{k-1})$ to compute $\boldsymbol{b}_i=\text{DReLU}(\boldsymbol{a}_i),i\in[k]$.
  \item Distribute the shares of $\llbracket \boldsymbol{b}\rrbracket_d$ to the parties.
\end{enumerate}
\end{tcolorbox}

\subsection{Proofs}
\setcounter{theorem}{0}
\setcounter{proposition}{0}
\label{appendixC}

\begin{proposition}
If $xy\leq 2^k$ for some $k<\ell$, then with probability at least $1-2^{k}/(2^{\ell}-1)$ it holds that $z = \lfloor xy/2^{\ell_x}\rfloor+v$ for some $v\in \{0, 1\}$.
\end{proposition}

\begin{proof}
    Let $z=xy,z'=(z+r)/2^{\ell_x}-r',$
    \begin{align}
            z'&=\frac{z+r-((z+r)\pmod{2^{\ell_x}})}{2^{\ell_x}}-r'\\
        &=\frac{z+r-(2^{\ell}-1)u-(z\pmod{2^{\ell_x}}+r\pmod{2^{\ell_x}})}{2^{\ell_x}}\\
        &+\frac{2^{\ell_x}v}{2^{\ell_x}}-r'\\
        &=\frac{z-z\pmod{2^{\ell_x}}}{2^{\ell_x}}+\frac{r-r\pmod{2^{\ell_x}}}{2^{\ell_x}}\\
        &-r'+v-\frac{(2^{\ell}-1)u}{2^{\ell_x}}\\
        &=z/2^{\ell_x}+r'+v-r'-\frac{(2^{\ell}-1)u}{2^{\ell_x}}\\
        &=z/2^{\ell_x}+v-2^{{\ell}-{\ell_x}}u,
        \end{align}
        where $u=((z+r)\geq p),p=2^\ell-1,v=((z'\pmod{2^d}+r'\pmod{2^{\ell_x}})\geq 2^{\ell_x}).$ The error is $v-2^{{\ell}-{\ell_x}}u$. The probability that $u=1$ is the probability that $r\geq p-z$. Moreover, since $z<2^{k}$ and $r$ is uniformly  random in $\mathbb{F}_p$, this probability is upper bounded by $2^{k}/(2^{\ell}-1)$.
\end{proof}

\begin{proposition}
A Shamir secret share $[x|_a]_d,$ where $f(a)=x,n=2d+1,a\notin\{1,2,\cdots,n\},f$ is a degree-$d$ polynomial. Each of the $n$ parties holds its own share $f(i),i\in\{1,2,\cdots, n\}$. The secret value $x$ can be locally converted to $[x|_b]_{2d}$, with the conversion performed as $f'(i)=f(i)\Pi_{j=1,j\neq i}^n\frac{a-j}{b-j}$.
\end{proposition}
\begin{proof}
        \begin{align}
            x &= \sum_{i=1}^nf(i)\Pi_{j=1,j\neq i}^n\frac{a-j}{i-j} =\sum_{i=1}^{t+1}f(i)\Pi_{j=1,j\neq i}^{t+1}\frac{a-j}{i-j} \\
                 &= \sum_{i=1}^{n}(\Pi_{j=1,j\neq i}^{n}\frac{b-j}{i-j})(f(i)\Pi^n_{j=1,j\neq i}\frac{a-j}{b-j}) \\
                 &= \sum_{i=1}^{n}(f(i)\Pi_{j=1,j\neq i}^n\frac{b-j}{i-j})(\Pi_{j=1,j\neq i}^n\frac{a-j}{b-j})=x  
        \end{align}
\end{proof}

For security proof, we only describe the simulation for protocol $\Pi_{VecMatMult}^{Fixed}$. The security proof for other protocols are simple compositions of local computations and invocations of ideal functionalities.

\begin{definition}[Semi-Honest Security]
Let $\Pi$ be a $n$-party protocol running in real-world and $\mathcal{F}:(\{0,1\}^{\lambda_1})^n\rightarrow(\{0,1\}^{\lambda_2})^n$ be the ideal randomized functionality. We say $\Pi$ securely computes $\mathcal{F}$
in the presence of a subset $\mathcal{C}\subset \{P_1,P_2,\cdots,P_n\}$ of the parties can be corrupted by semi-honest adversaries if for every corrupted $P_i,i\in \mathcal{C}$ and every input $\boldsymbol{x}\in(\{0,1\}^{\lambda_1})^n,$ there exists a simulator $\mathcal{S}$ such that: 
\begin{center}
    $\{view_{i,\Pi}(\boldsymbol{x}),output_{\Pi}(\boldsymbol{x})\}\overset{c}{\approx}\{\mathcal{S}(\boldsymbol{x}_i,\mathcal{F}_i(\boldsymbol{x})),\mathcal{F}(\boldsymbol{x})\}$,
\end{center}
where $view_{i,\Pi}(\boldsymbol{x})$ is the view of $P_i$ in the execution of $\Pi$ on $\boldsymbol{x}$, $output_{\Pi}(\boldsymbol{x})$ is the output of all parties and $\mathcal{F}_i(\boldsymbol{x})$ denotes the $i$-th output of $\mathcal{F}(\boldsymbol{x}).$
\end{definition}

\begin{theorem}
$\Pi_{VecMatMult}^{Fixed}$ securely realizes $\mathcal{F}_{VecMatMult}^{Fixed}$ in the $\mathcal{F}_{TruncTriple}$-hybrid model with abort, in the presence of a fully semi-honest adversary controlling $t$ corrupted parties.
\end{theorem}

\begin{proof}
    We will construct a simulator $\mathcal{S}$ to simulate the behaviors of honest parties. Let $\mathcal{C}$ denote the set of corrupted parties and $\mathcal{H}$ denote the set of honest parties. The simulator $\mathcal{S}$ works as follows.
    
     In the beginning, $\mathcal{S}$ receives the input shares of $\llbracket \boldsymbol{a}\rrbracket_d$ and $\llbracket A\rrbracket_d$ held by corrupted parties. When invoking $\mathcal{F}_{TruncTriple}$, $\mathcal{S}$ invokes the simulator of $\mathcal{F}_{TruncTriple}$ and receives from the adversary the shares of  $\llbracket \boldsymbol{r}\rrbracket_{2d}\in\mathbb{F}_p^v,\llbracket \boldsymbol{r}'\rrbracket_{d}\in \mathbb{F}_p^{\lceil v/k \rceil}$ held by corrupted parties.
    
    In Step 1, for each honest party, $\mathcal{S}$ samples a random vector as its shares of $\llbracket \boldsymbol{a}\rrbracket_d\cdot\llbracket A\rrbracket_d+\llbracket \boldsymbol{r}\rrbracket_{2d}.$ For each corrupted party, S computes its share of $\llbracket \boldsymbol{a}\rrbracket_d\cdot\llbracket A\rrbracket_d+\llbracket \boldsymbol{r}\rrbracket_{2d}.$  
    
    In Step 2, depending on whether $P_1$ is a corrupted party, there are two cases:
    \begin{itemize}
        \item If $P_1$ is an honest party, $\mathcal{S}$ uses these shares to reconstruct the secret $\boldsymbol{z}=\boldsymbol{a}\cdot A+\boldsymbol{r}$, and sends $\boldsymbol{z}$ back to corrupted parties. 
        \item If $P_1$ is a corrupted party, $\mathcal{S}$ sends the shares of  $\boldsymbol{z}=\boldsymbol{a}\cdot A+\boldsymbol{r}$ of honest parties to $P_1$. $P_1$ also receives the shares from corrupted parties. Then $P_1$ can reconstruct $\boldsymbol{z}$ and send $\boldsymbol{z}$ back to corrupted parties.
    \end{itemize}
     In Step 3, $\mathcal{S}$ computes the shares of $\boldsymbol{c}$ held by corrupted parties. Note that $\mathcal{S}$ has computed the shares of $\llbracket \boldsymbol{r}'\rrbracket_{d}$ held by corrupted parties when simulating $\mathcal{F}_{TruncTriple}$.

     In the real view, $\boldsymbol{z}$ is uniformly random since $\boldsymbol{r}$ is uniformly random. Hence, the view generated by the simulator $\mathcal{S}$ is indistinguishable from the real view underlying the security of $\mathcal{F}_{TruncTriple}$.
\end{proof}

\begin{theorem}
\label{RandomPairs}
$\Pi_{VM-RandTuple}$ securely realizes $\mathcal{F}_{VM-RandTuple}$ in the stand-alone model with abort, in the presence of a fully semi-honest adversary controlling $t$ corrupted parties.
\end{theorem}

\begin{proof} 
As for correctness, we need to generate the random pairs $(\llbracket\boldsymbol{r}\rrbracket_{2d},\llbracket\boldsymbol{r}'\rrbracket_{d}),$ where $\boldsymbol{r}=\{\boldsymbol{r}_i\}^{i=0}_{k^2-1},\llbracket\boldsymbol{r}\rrbracket_{2d}\in\mathbb{F}_{p}^{k},\boldsymbol{r}'=\{\sum_{j=i\cdot k}^{(i+1)\cdot k-1}\boldsymbol{r}_j\}_{i=0}^{k-1},\llbracket\boldsymbol{r}'\rrbracket_{d}\in\mathbb{F}_{p}$. $\llbracket\boldsymbol{r}\rrbracket_{2d}$ can be easily generated. First, $P_i$ randomly samples a vector $\boldsymbol{a}^i\in\mathbb{F}_{p}^{k^2}$ and sequentially pack every $k$ values into a degree-$2d$ PSS, i.e., $\llbracket\boldsymbol{a}^i\rrbracket_{2d}\in\mathbb{F}_{p}^{k}$ and distribute them to the other parties. Then, all parties can utilize a Vandermonde
matrix for the batch generation of such random shares. Every $k$ consecutive elements of the vector $\boldsymbol{r}'$ are sequentially packed into a PSS, i.e., $\llbracket \boldsymbol{r}'\rrbracket_d\in\mathbb{F}_{p}$. Each element of vector $\boldsymbol{r}'$ corresponds to the sum of $k$ values from vector $\boldsymbol{r}$, where $k$ values are respectively located at positions $\boldsymbol{s}_0$ to $\boldsymbol{s}_{k-1}$. Hence, $P_i$ first performs local summation on vector $\boldsymbol{a}^i$ for every $k$ values to obtain vector $\boldsymbol{b}^i$. Then $P_i$ perform a threshold-$(n, t)$ Shamir secret sharing sequentially on every $k$ values of vector $\boldsymbol{b}^i$, storing the secrets at corresponding positions in $\boldsymbol{s}_0$ to $\boldsymbol{s}_{k-1}$ in Step 2. Finally, parties combine the Shamir secret shares located at different positions into a PSS using a unit vector $\llbracket E_i\rrbracket_{k-1}$. Note that $\llbracket E_i\rrbracket_{k-1}$ is a degree-${k-1}$ PSS and each element of vector $\boldsymbol{f}^i$ is a degree-${t}$ Shamir secret share, so the resulting share has degree $d=t+k-1$, yielding $\llbracket \boldsymbol{r}'\rrbracket_d$. The privacy guarantees stem from Shamir secret sharing and packed Shamir secret sharing.
\end{proof}

\begin{theorem}
\label{theoremVector-Matrix}
$\Pi_{VecMatMult}$ securely realizes $\mathcal{F}_{VecMatMult}$ in the $\mathcal{F}_{VM-RandTuple}$-hybrid model with abort, in the presence of a fully semi-honest adversary controlling $t$ corrupted parties.
\end{theorem}

\begin{proof}
The correctness can be easily proved by the theory of block matrix multiplication. As for privacy, each $\boldsymbol{z}_i$ revealed in Step 2 leaks no information about $\boldsymbol{a}$ or $A$. As $\boldsymbol{r}$ is assumed to be uniformly random, $\boldsymbol{z}_i$ is a uniformly random.
\end{proof}

\begin{theorem}
\label{theoremVector-MatrixTrunc}
$\Pi_{VecMatMult}^{Fixed}$ securely realizes $\mathcal{F}_{VecMatMult}^{Fixed}$ in the $\mathcal{F}_{TruncTriple}$-hybrid model with abort, in the presence of a fully semi-honest adversary controlling $t$ corrupted parties.
\end{theorem}
\begin{proof} 
The correctness holds because $\boldsymbol{f}/2^{\ell_x}-\boldsymbol{r}'=(\boldsymbol{a}\cdot A)/2^{\ell_x}$. The privacy can be proved similarly as Theorem~\ref{theoremVector-Matrix}. 
\end{proof}

\begin{theorem}
\label{theoremTruncTriple}
$\Pi_{TruncTriple}$ securely realizes $\mathcal{F}_{TruncTriple}$ in the $(\mathcal{F}_{RandomBits},\mathcal{F}_{Random},\mathcal{F}_{DegreeT rans})$-hybrid model with abort, in the presence of a fully semi-honest adversary controlling $t$ corrupted parties.
\end{theorem}

\begin{proof} 
As for correctness, there
is the following relation: $\boldsymbol{q}_i/2^{\ell_x}=\boldsymbol{r}'_i,i\in[k],\boldsymbol{q}=\sum_{i=0}^{k}\boldsymbol{w}^i,\boldsymbol{q}_i=\sum_{j=0}^k\boldsymbol{w}_i^j$. Each $\boldsymbol{w}^i$ is a PSS at the positions $\boldsymbol{s}_0$ to $\boldsymbol{s}_{k-1}$. We can get $\llbracket \boldsymbol{r}\rrbracket_{2d}$ by Step 4 and 5. The privacy follows from the fact that this protocol invokes the private ideal functionalities $\mathcal{F}_{RandomBits},\mathcal{F}_{Random},\mathcal{F}_{DegreeTrans}$.
\end{proof}

\begin{theorem}
\label{theoremPMatMult}
$\Pi_{PMatMult}^{Fixed}$ securely realizes $\mathcal{F}_{PMatMult}^{Fixed}$ in the $\mathcal{F}_{RandomBits}$-hybrid model with abort, in the presence of a fully semi-honest adversary controlling $t$ corrupted parties.
\end{theorem}

\begin{proof}
The correctness can be proved similarly as the DN protocol~\cite{damgaard2007scalable}.
The privacy can be proved similarly as Theorem~\ref{theoremVector-Matrix}.
\end{proof}

\begin{theorem}
\label{theoremPackTrans}
$\Pi_{PackTrans}$ securely realizes $\mathcal{F}_{PackTrans}$ in the $(\mathcal{F}_{DegreeTrans},\mathcal{F}_{Random})$-hybrid model with abort, in the presence of a fully semi-honest adversary controlling $t$ corrupted parties.
\end{theorem}

\begin{proof}
The correctness holds because $\boldsymbol{x}+\boldsymbol{r}=(\boldsymbol{x}_0+\boldsymbol{r}_0,\cdots,\boldsymbol{x}_{k-1}+\boldsymbol{r}_{k-1}),\boldsymbol{z}^i=(\boldsymbol{x}_{i}+\boldsymbol{r}_{i},\cdots,\boldsymbol{x}_{i}+\boldsymbol{r}_{i}),\boldsymbol{z}^i-\boldsymbol{r}'^i=(\boldsymbol{x}_{i},\cdots,\boldsymbol{x}_{i})$.
The privacy can be proved similarly as Theorem~\ref{theoremVector-Matrix}.
\end{proof}

\begin{theorem}
\label{theoremPreOR}
$\Pi_{PreOR}$ securely realizes $\mathcal{F}_{PreOR}$ in the $\mathcal{F}_{PreMult}$-hybrid model with abort, in the presence of a fully semi-honest adversary controlling $t$ corrupted parties.
\end{theorem}

\begin{proof}
The correctness holds because computing $\boldsymbol{b}^i_j=\lor_{v=0}^{j}\boldsymbol{a}^i_v$ is equivalent to compute $\bar{\boldsymbol{b}}_j^i=\land_{v=0}^{j}\bar{\boldsymbol{a}}_v^i$ where $\bar{\boldsymbol{a}}_v^i$ denote the opposite bit of $\boldsymbol{a}_v^i$ and so is $\boldsymbol{b}_j^i$. The privacy follows from the fact that this protocol only invokes the private ideal functionality $\mathcal{F}_{PreMult}$.
\end{proof}

\begin{theorem}
\label{theoremBitwiseLT}
$\Pi_{Bitwise-LT}$ securely realizes $\mathcal{F}_{Bitwise-LT}$ in the $(\mathcal{F}_{PreOR},\mathcal{F}_{DegreeT rans},\mathcal{F}_{Xor})$-hybrid model with abort, in the presence of a fully semi-honest adversary controlling $t$ corrupted parties.
\end{theorem}

\begin{proof}
The correctness can be proved similarly as LXY24. The privacy follows from the fact that this protocol invokes the private ideal functionalities $\mathcal{F}_{Xor},\mathcal{F}_{PreOR},\mathcal{F}_{DegreeTrans}$.
\end{proof}

\begin{theorem}
\label{theoremDReLU}
$\Pi_{DReLU}$ securely realizes $\mathcal{F}_{DReLU}$ in the $(\mathcal{F}_{Bitwise-LT}, \mathcal{F}_{Xor})$-hybrid model with abort, in the presence
of a fully semi-honest adversary controlling $t$ corrupted parties.
\end{theorem}

\begin{proof}
The correctness follows easily from $\text{LSB}(2\boldsymbol{a})=\text{LSB}(\boldsymbol{y})\oplus\text{LSB}(\boldsymbol{r})\oplus(\boldsymbol{y}<\boldsymbol{r})$. As for privacy, the revealed vector $\boldsymbol{y}$ leaks no information about $\boldsymbol{a}$, as $\boldsymbol{r}$ is uniformly random. Then the invocation of the ideal functionalities $\mathcal{F}_{Xor}$ and $\mathcal{F}_{Bitwise-LT}$ is private.
\end{proof}

\begin{theorem}
\label{theoremReLU}
$\Pi_{ReLU}$ securely realizes $\mathcal{F}_{ReLU}$ in the $(\mathcal{F}_{DReLU},\mathcal{F}_{PDN})$-hybrid model with abort, in the presence of a fully semi-honest adversary controlling $t$ corrupted parties.
\end{theorem}

\begin{proof}
The correctness follows easily from $\text{ReLU}(\boldsymbol{a})=\text{DReLU}(\boldsymbol{a})\cdot\boldsymbol{a}$. The privacy follows from the fact that this protocol invokes the private ideal functionalities $\mathcal{F}_{DReLU},\mathcal{F}_{PDN}$.
\end{proof}

\begin{theorem}
\label{theoremM axpool}
$\Pi_{Maxpool}$ securely realizes $\mathcal{F}_{Maxpool}$ in the $\mathcal{F}_{ReLU}$-hybrid model with abort, in the presence of a fully semi-honest adversary controlling $t$ corrupted parties.
\end{theorem}

\begin{proof}
The correctness and privacy can be proved similarly as LXY24.
\end{proof}

\bibliographystyle{IEEEtran}
\bibliography{IEEEabrv,references}

\end{document}